\documentclass[a4paper,UKenglish,cleveref, autoref, thm-restate,numberwithinsect]{lipics-v2021}

\nolinenumbers

\usepackage{xspace}
\usepackage{makecell}
\usepackage{caption}
\usepackage{subcaption}
\usepackage{booktabs}
\usepackage{bm}
\usepackage{complexity}
\usepackage{tikz}
\usepackage{pgfplots}
\usetikzlibrary{automata}

\usepackage{amssymb}
\usepackage{mathtools}

\usetikzlibrary{arrows,automata,chains,matrix,positioning,scopes,calc,decorations.pathmorphing,shapes.misc,shapes,fit}
\tikzset{mode/.style={font=\scriptsize}}
\tikzset{gadget/.style={->,>=stealth,initial text=,minimum size=7pt,auto,on grid,scale=1,inner sep=1pt,node distance=1cm}}
\tikzset{gadgeto/.style={-,>=stealth,initial text=,minimum size=7pt,auto,on grid,scale=1,inner sep=1pt,node distance=1cm}}
\tikzset{every state/.style={minimum size=15pt,inner sep=1pt,fill=black!10,draw=black!70,thick}}

\newcommand{\eg}{\textit{e.g.}\@\xspace}
\newcommand{\ie}{\textit{i.e.}\@\xspace}

\newcommand{\N}{\mathbb{N}}
\newcommand{\Z}{\mathbb{Z}}
\newcommand{\set}[1]{\{#1\}}

\newcommand{\step}{\to}
\newcommand{\steps}{\xrightarrow{*}}
\newcommand{\zstep}{\underset{\Z}{\to}}
\newcommand{\zsteps}{\underset{\Z}{\overset{*}\rightarrow}}
\newcommand{\weakstep}{\Rightarrow}
\newcommand{\weaksteps}{\xRightarrow{*}}
\newcommand{\Succs}{\mathsf{Succ}}

\newcommand{\xsteps}[1]{\xrightarrow{#1}}
\newcommand{\xzsteps}[1]{\underset{\Z}{\overset{#1}\rightarrow}}
\newcommand{\xweaksteps}[1]{\xRightarrow{#1}}

\newcommand {\weak} {\text{monus}}
\newcommand{\copyv}[1]{\mathsf{copy}(#1)}
\newcommand{\copyi}[1]{\mathsf{copy}_i(#1)}
\newcommand{\extend}[1]{\mathsf{extend}(#1)}

\newcommand{\oomit}[1]{}

\newcommand{\rev}[1]{{#1}^{\mathsf{rev}}}

\newcommand{\mi}{\text{min}_\Z}

\newcommand{\vv}{\mathbf{v}}
\newcommand{\ww}{\mathbf{w}}
\newcommand{\zz}{\mathbf{z}}
\newcommand{\mm}{\mathbf{m}}
\newcommand{\uu}{\mathbf{u}}
\newcommand{\ee}{\mathbf{e}}

\newcommand{\cV}{\mathcal{V}}
\newcommand{\cA}{\mathcal{A}}
\newcommand{\bzero}{\bm{0}}

\title{Monus semantics in vector addition systems with states}

\author{Pascal Baumann}{Max Planck Institute for Software Systems (MPI-SWS), Germany}{pbaumann@mpi-sws.org}{https://orcid.org/0000-0002-9371-0807}{}

\author{Khushraj Madnani}{Max Planck Institute for Software Systems (MPI-SWS), Germany}{kmadnani@mpi-sws.org}{https://orcid.org/0000-0003-0629-3847}{}

\author{Filip Mazowiecki}{University of Warsaw, Poland}{f.mazowiecki@mimuw.edu.pl}{https://orcid.org/0000-0002-4535-6508}{Supported by the ERC grant INFSYS, agreement no. 950398.}

\author{Georg Zetzsche}{Max Planck Institute for Software Systems (MPI-SWS), Germany}{georg@mpi-sws.org}{https://orcid.org/0000-0002-6421-4388}{}

\authorrunning{P. Baumann, K. Madnani, F. Mazowiecki, and G. Zetzsche}

\Copyright{Pascal Baumann, Khushraj Madnani, Filip Mazowiecki, and Georg Zetzsche}



\acknowledgements{The authors are grateful to Wojciech Czerwi\'{n}ski and Sylvain Schmitz for discussions, and to Sylvain Schmitz for suggesting the term 'monus'.}

\EventEditors{Guillermo A. P\'{e}rez and Jean-Fran\c{c}ois Raskin}
\EventNoEds{2}
\EventLongTitle{34th International Conference on Concurrency Theory (CONCUR 2023)}
\EventShortTitle{CONCUR 2023}
\EventAcronym{CONCUR}
\EventYear{2023}
\EventDate{September 18--23, 2023}
\EventLocation{Antwerp, Belgium}
\EventLogo{}
\SeriesVolume{279}
\ArticleNo{10}

\begin{document}
	
	\maketitle
	
	\begin{abstract}
		Vector addition systems with states (VASS) are a popular model for concurrent systems. However, many decision problems have prohibitively high complexity. Therefore, it is sometimes useful to consider overapproximating semantics in which these problems can be decided more efficiently.
		
		We study an overapproximation, called monus semantics, that slightly relaxes the semantics of decrements: A key property of a vector addition systems is that in order to decrement a counter, this counter must have a positive value. In
		contrast, our semantics allows decrements of zero-valued counters: If such a
		transition is executed, the counter just remains zero.
		
		It turns out that if only a subset of transitions is used with monus semantics
		(and the others with classical semantics), then reachability is undecidable.
		However, we show that if monus semantics is used throughout, reachability
		remains decidable. In particular, we show that reachability for VASS with monus semantics is as hard as that of classical VASS (i.e. Ackermann-hard), while the zero-reachability and coverability are easier (i.e. $\EXPSPACE$-complete and $\NP$-complete, respectively).
		We provide a comprehensive account of the complexity of the
		general reachability problem, reachability of zero configurations, and
		coverability under monus semantics. We study these problems in general VASS,
		two-dimensional VASS, and one-dimensional VASS, with unary and binary counter
		updates.
	\end{abstract}

	\section{Introduction}
	Vector addition systems with states (VASS) are an established model used in formal verification with a wide range of applications, \eg in concurrent systems~\cite{GermanS92}, business processes~\cite{AL97} and others (see the survey~\cite{Schmitz16}). They are finite automata with transitions labeled by vectors over integers in some fixed dimension $d$. A configuration of a VASS consists of a pair $(p, \vv)$, denoted $p(\vv)$, where $p$ is a state and $\vv$ is a vector in $\N^d$. As a result of applying a transition labeled by some $\zz \in \Z^d$, the vector in the resulting configuration is $\vv + \zz$. Thus in particular $\vv + \zz \ge \bzero$ must hold for the transition to be applicable. The latter requirement is often called the VASS semantics. To avoid ambiguity we will refer to it as the \emph{classical VASS semantics}.
	
	The VASS model is also studied with other semantics.
	One of the most natural variants of VASS semantics is the \emph{integer semantics} (or simply \emph{$\Z$-semantics}), where configurations are of the form $p(\vv)$, where $\vv \in \Z^d$~\cite{DBLP:conf/rp/HaaseH14}. There, a transition can always be applied, \ie the resulting configuration is $\vv + \zz$ and we do not require $\vv + \zz \ge \bzero$. In this paper we consider VASS with the \emph{monus semantics}, whose behavior partly resembles both classical and integer semantics. There, a transition can always be applied (as in $\Z$-semantics), however, if as a result the vector in the new configuration would have negative entries, then these are replaced with $0$. Thus, vectors in configurations are over the naturals (as in classical semantics). The name monus semantics comes from the monus binary operator, which is a variant of the minus operator.\footnote{One can also think that monus semantics is integer semantics, where after every step we apply the ReLU function.} Note that every instance of a VASS can be considered with all three semantics. See \cref{fig:intro_semantics} for an example.
	
	\begin{figure}
		\centering
		\begin{tikzpicture}
			\node[state] (p) {$p$};
			
			\node[above right = -0.7cm and 1.5cm  of p] (q) {
				\begin{tabular}{l l}
					\text{classical} & $p(2,0) \step p(1,2) \step p(0,4) \not \step$ \\
					\text{integer} & $p(2,0) \zstep p(1,2) \zstep p(0,4) \zstep p(-1,6) \zsteps p(-n,4+2n)$ \\
					\text{monus} & $p(2,0) \weakstep p(1,2) \weakstep p(0,4) \weakstep p(0,6) \weaksteps p(0,4+2n)$
				\end{tabular}
			};
			
			\path
			(p) edge[loop above] node {$(-1,2)$} (p)
			;
		\end{tikzpicture}
		\caption{A VASS in dimension $2$ with one state $p$ and one transition $t$. It has only one transition labeled with $(-1,2)$. We consider possible runs assuming that the initial configuration is $p(2,0)$. We use different notation for steps in each semantics: $\step$, $\zstep$, $\weakstep$. For the classical semantics ($\step$) after reaching the configuration $p(0,4)$ the transition can no longer be applied. For the integer semantics ($\zsteps$) the transition can be applied even in $p(0,4)$, reaching all configurations of the form $p(-n,4+2n)$. Similarly for the monus semantics ($\weaksteps$), but there the configurations reachable from $p(0,4)$ are of the form $p(0,4+2n)$.}\label{fig:intro_semantics}
	\end{figure}
	
	We study classical decision problems for VASS: reachability and coverability. The input for these problems is a VASS $\cV$, an initial configuration $p(\vv)$, and a final configuration $q(\ww)$. The \emph{reachability} problem asks whether there is a run from $p(\vv)$ to $q(\ww)$. A variant of this problem, called \emph{zero reachability}, requires additionally that in the input the final vector is fixed to $\ww = \bzero$. The \emph{coverability} problem asks whether there is a run from $p(\vv)$ to $q(\ww')$, where $\ww' \ge \ww$. Note that all three problems can be considered with respect to any of the three VASS semantics. As an example consider the VASS in \cref{fig:intro_semantics}. Then for all three semantics $p(1,2)$ is both reachable and coverable from $p(2,0)$; and $p(0,2)$ is not reachable from $p(2,0)$ (but it is coverable as $(1,2) \ge (0,2)$).
	
	\subparagraph*{Contribution I: Arbitrary dimension.}
	Our first contribution is settling the complexities of reachability and coverability for VASS with the monus semantics (see \cref{table-results}).
	We prove that reachability is Ackermann-complete by showing that it is inter-reducible with classical VASS reachability, which is known to be~Ackermann-complete~\cite{DBLP:conf/lics/LerouxS19,DBLP:conf/focs/CzerwinskiO21,DBLP:conf/focs/Leroux21}.
	This comes as a surprise, since in monus semantics, every transition can always be applied, just like in $\Z$-semantics,
	where reachability is merely $\NP$-complete~\cite{DBLP:conf/rp/HaaseH14}. Thus, the monus operation encodes enough information in the resulting configuration that reachability remains extremely hard.
	
	The Ackermann-hardness relies crucially on the fact that the final configuration is non-zero: We also show that the zero reachability problem is $\EXPSPACE$-complete in monus semantics. This uses inter-reducibility with classical VASS coverability, which is $\EXPSPACE$-complete due to seminal results of Lipton and Rackoff~\cite{lipton1976reachability,Rackoff78}. The fact that zero-reachability is significantly easier than general reachability is in contrast to classical semantics, where zero reachability is interreducible with the reachability problem (intuitively, one can modify the input VASS by adding an extra edge that decrements by $\ww$).
	
	In another unexpected result, the complexity of coverability drops even more: We prove that it is $\NP$-complete in monus semantics. We complete these results by showing that mixing classical and monus semantics (\ie each transition is designated to either work in classical or monus semantics) makes reachability undecidable.
	
	\subparagraph*{Contribution II: Fixed dimension.}
	Understanding the complexity of reachability problems in VASS of fixed dimension has received a lot of attention in recent years and is now well understood. This motivates our second contribution: An almost complete complexity analysis of reachability, zero reachability and coverability for VASS with the monus semantics in dimensions $1$ and $2$. Here, the complexity depends on whether the counter updates are encoded in unary or binary (see \cref{table-results}). 

	We restrict our attention to dimensions $1$ and $2$ as most research in fixed dimension for the classical semantics.
	For the classical semantics not much is known about reachability in dimension $d\ge 3$. Essentially, the only known results consist of an upper bound of $\mathbf{F}_7$ that follows from the Ackermann upper bound in the general case~\cite{DBLP:conf/lics/LerouxS19}, and a $\PSPACE$-lower bound that holds already for $d=2$~\cite{DBLP:journals/jacm/BlondinEFGHLMT21}. An intuition as to why the jump from $2$ to $3$ is so difficult is provided already by Hopcroft and Pansiot~\cite{hopcroft1979reachability} who prove that the reachability set is always semilinear in dimension $2$, and show an example that this is not the case in dimension~$3$. In contrast, coverability is well understood, and already Rackoff's construction~\cite{Rackoff78} shows that for fixed dimension $d \ge 2$ coverability is in $\NL$ and in $\PSPACE$, for unary and binary encoding, respectively (with matching lower bounds~\cite{DBLP:journals/jacm/BlondinEFGHLMT21}).
	
	\subparagraph*{Key technical ideas} The core insights of our paper are
	characterizations of the reachability and coverability relations in monus
	semantics, in terms of reachability and coverability in classical and
	$\Z$-semantics
	(\cref{charact-reachability,charact-coverability,charact-zero-reachability}).
	These allow us to apply a range of techniques to reduce reachability problems
	for one semantics into problems for other semantics, and thereby transfer
	existing complexity results. There are three cases where we were unable to
	ascertain the exact complexity: (i)~reachability in $2$-VASS with unary counter
	updates, (ii)~zero reachability in $1$-VASS with binary updates, and
	(iii)~coverability in $1$-VASS with binary counter updates.  Concerning (i),
	this is because for $2$-VASS with unary updates, it is known that classical
	reachability is $\NL$-complete~\cite{DBLP:journals/jacm/BlondinEFGHLMT21}, but
	we would need to decide existence of a run that visits intermediate
	configurations of a certain shape.  In the case of $2$-VASS with binary
	updates, the methods from~\cite{DBLP:journals/jacm/BlondinEFGHLMT21} (with a
	slight extension from~\cite{BaumannMeyerZetzsche2022a}) allow this. The other cases,
	(ii) and (iii), are quite similar to each other.
	In particular, problem (ii) is logspace-interreducible
	with classical coverability in $1$-VASS with binary updates, for which only
	an $\NL$ lower bound and an $\NC^2$ upper bound are known~\cite{DBLP:conf/concur/AlmagorCPS020}.

	\subparagraph*{Monus semantics as an overapproximation.}
	Recall the example in \cref{fig:intro_semantics}. Notice that every configuration reachable in the classical semantics is also reachable in the integer and monus semantics. It is not hard to see that this is true for every VASS model. Such semantics are called \emph{overapproximations} of the classical VASS semantics. Overapproximations are a standard technique used in implementations of complex problems, in particular for the VASS model (see the survey~\cite{Blo20}). They allow to prune the search space of reachable configurations, based on the observation that if a configuration is not reachable by an overapproximation then it cannot be reachable in the classical semantics. This is the core idea behind efficient implementations both of the coverability problem~\cite{ELMMN14,BFHH16} and the reachability problem~\cite{DixonL20,BlondinHO21}.
	
	The two most popular overapproximations, integer semantics~\cite{DBLP:conf/rp/HaaseH14} and continuous\linebreak semantics~\cite{FH15}, behave similarly for both reachability and coverability problems, namely both problems are $\NP$-complete. Note that all of the implementations mentioned above rely on such algorithms in $\NP$ as they can be efficiently implemented via SMT solvers. Interestingly, the monus semantics is an efficient overapproximation only for the coverability problem. (As far as we know this is the first study of a VASS overapproximation with this property.) Therefore, it seems to be a promising approach to try to speed up backward search algorithms using monus semantics (in the same vein as~\cite{BFHH16}). Whether this leads to improvements in practice remains to be seen in future work.
	
	\newclass{\TWOEXPTIME}{2EXPTIME}
	
	\subparagraph*{Related work.}
	We discuss related work for VASS in classical semantics.
	A lot of research is dedicated to reachability for the flat VASS model, \ie a model that does not allow for nested cycles in runs. In dimension $2$ decision problems for VASS reduce to flat VASS, which is crucial to obtain the exact complexities~\cite{DBLP:journals/jacm/BlondinEFGHLMT21}. It is known that in dimensions $d \ge 3$ such a reduction is not possible, but this raised natural questions of the complexity for flat VASS in higher dimensions~\cite{Czerwinski0LLM20,CzerwinskiO22}. Another research direction is treating the counters in VASS models asymmetrically. For example, it is known that allowing for zero tests in VASS makes reachability and coverability undecidable (they essentially become Minsky machines). However, it was shown that if only one of the $2$ counters is allowed to be zero tested then both reachability and coverability remain $\PSPACE$-complete~\cite{LerouxSutre2020a}. A different asymmetric question is when one counter is encoded in binary and the other is encoded in unary. Then recently it was shown that coverability is in $\NP$~\cite{MSW23} but it is unknown whether there is a matching lower bound. Finally, there are two important extensions of the VASS model: branching VASS (where runs are trees, not paths), and pushdown VASS (with one pushdown stack). For branching VASS, coverability is $\TWOEXPTIME$-complete~\cite{DemriJLL13}. The complexity of reachability is well understood in dimension $1$~\cite{GollerHLT16,FigueiraLLMS17} but in dimension $2$ or higher it is unknown whether it is decidable. For pushdown VASS only coverability in dimension $1$ is known to be decidable~\cite{LerouxST15}, otherwise decidability of both reachability and coverability remain open problems. Recently some progress was made on restricted pushdown VASS models~\cite{EnglertHLLLS21,GanardiMPSZ22}. The monus semantics is a natural overapproximation that can be studied in all of these variants. Finally, let us mention that VASS with monus semantics fit into the very general framework of G-nets~\cite{DBLP:conf/icalp/DufourdFS98}, but does not seem to fall into any of the decidable subclasses studied in~\cite{DBLP:conf/icalp/DufourdFS98}. However, if we equip VASS with with the usual well-quasi ordering on configurations, it is easy to see that even with monus semantics, they constitute well-structured transition systems (WSTS)~\cite{DBLP:journals/tcs/FinkelS01,DBLP:conf/lics/AbdullaCJT96}, which makes available various algorithmic techniques developed for WSTS.
	
	\subparagraph*{Organization.} In \cref{sec:vass} we formally define the VASS model and the classical, integer and monus semantics. In \cref{sec:arbitrary_dim} we prove the results in arbitrary dimension. Then in \cref{sec:two_dim} and \cref{sec:one_dim} we prove the results in dimension $2$ and $1$, respectively.

	\section{Vector addition systems with monus semantics: Main results}\label{sec:vass}
	Given a vector $\vv \in \Z^d$ we write $\vv[i]$ for the value in the $i$-th coordinate, where $i \in \set{1,\ldots, d}$.  We also refer to $i$ as the $i$-th counter and write that it contains $\vv[i]$ tokens. 
	Given two vectors $\vv$ and $\vv'$ we write $\vv \ge \vv'$ if $\vv[i] \ge \vv'[i]$ for all $i = 1,\ldots, d$.
	By $\mathbf{0}^d$ we denote the zero vector in dimension $d$. We also simply write $\bzero$ if $d$ is clear from context.
	
	\subparagraph*{Vector addition systems with states.}
	A vector addition system with states (VASS) is a triple $\cV=(d, Q,\Delta)$, where $d \in \N$, $Q$ is a finite set of states and $\Delta \subseteq Q \times \Z^d \times Q$ is a finite set of transitions.
	Throughout the paper we fix a VASS  $\cV = (d, Q,\Delta)$.
	
	We start with the formal definitions in the \emph{classical semantics}. A configuration of a VASS is a pair $p(\vv) \in Q \times \N^d$, denoted $p(\vv)$. Any transition $t \in \Delta$ induces a successor (partial) function $\Succs_t: Q \times \N^d \to Q \times \N^d$ such 
	that $\Succs_t (q(\vv)) = q'(\vv')$   iff $t = (q,\zz,q')$ and $\vv' = \vv + \zz$. This successor function can be lifted up to $\Delta$ to get a  step relation $\step_\cV$, such that any pair of configuration $C \step_\cV C'$ iff there exists $t \in \Delta$ with $\Succs_t(C)=C'$. A \emph{run} is a sequence of configurations
	\begin{align*}
		q_0(\vv_0), q_1(\vv_1), q_2 (\vv_2), \ldots, q_k (\vv_k)
	\end{align*}
	such that for every $0 < j \le k$, $q_{j-1}(\vv_{j-1}) \step_V q_{j}(\vv_{j})$. If there exists such a run we say that $q_k(\vv_k)$ is reachable from $q_0(\vv_0)$ and denote it $C_0 \mathrel{{\steps}{}_{\cV}} C_k$. We call $\mathrel{{\steps}{}_{\cV}}$ the reachability relation in the classical VASS semantics.
	
	In this paper we consider two additional semantics. The first is called the \emph{integer semantics} (or \emph{$\Z$-semantics}). A configuration in this semantics is a pair $p(\vv) \in Q \times \Z^d$ (hence, values of vector coordinates can drop below zero). The definitions of  successor function,  step relation and  run are analogous as for the classical semantics. By $\mathrel{{\zstep}{}_{\cV}}$  and $\mathrel{{\zsteps}{}_{\cV}}$, we denote the step  and reachability relations in the $\Z$-semantics, respectively.
	
	The second is called \emph{monus semantics}. The configurations are the same as in the classical semantics. The difference is in the successor function. Every transition $t \in \Delta$ induces a successor function $\Succs_t: Q \times \N^d \to Q \times \N^d$ as follows: 
	$\Succs_t(q(\vv)) = q'(\vv')$  iff $t = (q,\zz,q')$ and for all $j \in \{1,2, \ldots d\}$, $\vv'[j] = \max (\vv[j] + \zz[j], 0)$. We write in short $\vv' = \max(\vv+\zz,\bzero)$. Step relation and runs are defined analogously as in the case of classical semantics. By $\mathrel{{\weakstep}{}_{\cV}}$ and $\mathrel{{\weaksteps}{}_{\cV}}$, we denote the step and reachability relations in the monus semantics, respectively.
	
	We drop the subscript $\cV$ from the above relations when the VASS is clear from context. We write that a run is a \emph{classical run}, a \emph{$\Z$ run} or a \emph{monus run} to emphasize the considered semantics. An example highlighting the differences between the three semantics is in \cref{fig:intro_semantics}.

	\subparagraph*{Decision problems.}
	We study the following decision problems for VASS. 
	
	\smallskip
	
	The classical \emph{reachability problem}:
	\begin{description}
		\item[Given] A VASS $\cV=(d,Q,\Delta)$ and two configurations $p(\vv)$ and $q(\ww)$.
		\item[Question] Does $p(\vv) \weaksteps q(\ww)$ hold?
	\end{description}
	
	The classical \emph{zero reachability problem}:
	\begin{description}
		\item[Given] A VASS $\cV=(d,Q,\Delta)$, a configuration $p(\vv)$ and a state $q$.
		\item[Question] Does $(p, \vv) \weaksteps q(\mathbf{0}^d)$ hold?
	\end{description}
	
	The classical \emph{coverability problem}:
	\begin{description}
		\item[Given] A VASS $\cV=(d,Q,\Delta)$ and two configurations $p(\vv)$ and $q(\ww)$.
		\item[Question] Does $p(\vv) \weaksteps q(\ww')$ hold for some $\ww' \ge \ww$?
	\end{description}
	
	Similarly, the above problems in $\Z$ and classical semantics are defined by replacing $\weaksteps$ with $\zsteps$ and $\steps$, respectively.

	\subparagraph{Main results}
	The main complexity results of this work are summarized in
	\cref{table-results}.  In \cref{table-classical}, we recall complexity
	results for VASS with classical semantics for comparison.
	We do not split the cases of unary and binary encoding for arbitrary dimensions,
	since there all lower bounds work for unary, whereas all upper bounds work for binary.
	
	Concerning the \emph{reachability problem}, we note that in all cases where we
	obtain the exact complexity, it is the same as for the classical VASS
	semantics. For the other decision problems, there are stark differences: First,
	while in the classical semantics, \emph{zero reachability} is easily
	inter-reducible with general reachability, in the monus semantics, its
	complexity drops in two cases: In $1$-VASS with binary counter updates, monus
	zero reachability is in $\NC^2$ (thus polynomial time), compared to $\NP$ in the
	classical setting. Moreover, in arbitrary dimension, monus zero reachability is
	$\EXPSPACE$-complete, compared to Ackermann in the classical
	semantics.
	For the \emph{coverability problem}, the monus semantics also lowers the
	complexity in two cases: For binary encoded $2$-VASS ($\NP$ in monus semantics,
	$\PSPACE$ in classical) and in the general case ($\NP$ in monus semantics,
	$\EXPSPACE$ in classical semantics).

	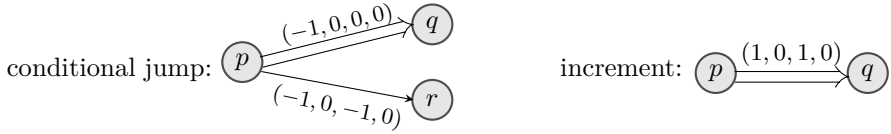
\begin{figure} 
		\text{conditional jump:}~\begin{tikzpicture}[gadgeto, node distance=2cm, baseline={([yshift=-.8ex]current bounding box.center)},
			Double/.style={
				to path={
					($(\tikztostart)!2pt!90:($(\tikztotarget)!5pt!(\tikztostart)$)$) -- ($($(\tikztotarget)!3.5pt!(\tikztostart)$)!2pt!270:(\tikztostart)$)
					($(\tikztostart)!2pt!270:($(\tikztotarget)!4pt!(\tikztostart)$)$) -- ($($(\tikztotarget)!3.5pt!(\tikztostart)$)!2pt!90:(\tikztostart)$)
					($($(\tikztotarget)!5pt!(\tikztostart)$)!4pt!90:(\tikztostart)$) 
					.. controls
					($($(\tikztotarget)!3pt!(\tikztostart)$)!0.5pt!90:(\tikztostart)$) and
					($(\tikztotarget)!1pt!(\tikztostart)$)
					.. (\tikztotarget) 
					.. controls
					($(\tikztotarget)!0.5pt!(\tikztostart)$) and
					($($(\tikztotarget)!3pt!(\tikztostart)$)!1pt!270:(\tikztostart)$) 
					..
					($($(\tikztotarget)!5pt!(\tikztostart)$)!4pt!270:(\tikztostart)$)
				}
			}
			]
			\node[state] (p) at (0,0) {$p$};
			\node[state] (q) at (2.5, 0.5) {$q$};
			\node[state] (r) at (2.5, -0.5) {$r$};
			\node[rotate =12] at (1.25, 0.5){\footnotesize{$(-1, 0, 0, 0)$}};
			\draw(0.25,0) to[Double](2.25,0.5);
			\draw[->](0.25,-0.125)--(2.25,-0.5);
			\node[rotate =-12] at (1.25, -0.6){\footnotesize{$(-1, 0, -1, 0)$}};
		\end{tikzpicture}~ ~~~~~~~~~~\text{increment}:
		\begin{tikzpicture}[gadgeto, node distance=1cm, baseline={([yshift=-.8ex]current bounding box.center)}, Double/.style={
				to path={
					($(\tikztostart)!2pt!90:($(\tikztotarget)!5pt!(\tikztostart)$)$) -- ($($(\tikztotarget)!3.5pt!(\tikztostart)$)!2pt!270:(\tikztostart)$)
					($(\tikztostart)!2pt!270:($(\tikztotarget)!4pt!(\tikztostart)$)$) -- ($($(\tikztotarget)!3.5pt!(\tikztostart)$)!2pt!90:(\tikztostart)$)
					($($(\tikztotarget)!5pt!(\tikztostart)$)!4pt!90:(\tikztostart)$) 
					.. controls
					($($(\tikztotarget)!3pt!(\tikztostart)$)!0.5pt!90:(\tikztostart)$) and
					($(\tikztotarget)!1pt!(\tikztostart)$)
					.. (\tikztotarget) 
					.. controls
					($(\tikztotarget)!0.5pt!(\tikztostart)$) and
					($($(\tikztotarget)!3pt!(\tikztostart)$)!1pt!270:(\tikztostart)$) 
					..
					($($(\tikztotarget)!5pt!(\tikztostart)$)!4pt!270:(\tikztostart)$)
				}
			}
			]
			\node[state] (p) at (0,0) {$p$};
			\node[state] (q) at (2, 0) {$q$};
			\node at (1, 0.25) {\small{$(1,0,1,0)$}};
			\draw (0.25,-0.05) to[Double] (1.75,-0.05);
			
		\end{tikzpicture}

		\caption{Two gadgets for realizing a zero-testable counter.}\label{undecidability}
	\end{figure}

	\subparagraph{Undecidability} To stress the subtle effects of
	monus semantics, we mention that it leads to undecidability if combined
	with classical semantics: If one can specify the applied semantics
	(classical vs.\ monus) for each transition, then (zero) reachability
	becomes undecidable. 
	
	We sketch the proof using \cref{undecidability}. It shows two gadgets, where ``$\to$''
	transitions use classical semantics and ``$\Rightarrow$'' transitions
	use monus semantics. The two gadgets realize a counter with zero test:
	The left gadget is a conditional jump (``if zero, then go to $q$,
	otherwise decrement and go to $r$''), whereas the right gadget is just
	an increment. In intended runs (i.e.\ where the left gadget always takes
	the intended transition), the counter value is stored both in
	components $1$ and $3$.  (To realize a full two-counter machine, the
	same gadgets on components $2$ and $4$ realize the other testable
	counter.) Thus, initially, all components are zero. Note that if the
	left gadget always takes the transitions as intended, then the first
	and third counter will remain equal. If the gadget takes the upper
	transition when the counter is not actually zero, then the first
	counter becomes smaller than the third, and will then always stay
	smaller. Hence, to reach $(0,0,0,0)$, the left gadget must
	always behave as intended.

However, coverability remains decidable if we can specify the semantics of each
transition. Indeed, suppose we order the configurations of a VASS by the usual
well-quasi ordering (i.e.\ the control states have to agree, and the counter
values are ordered component-wise). Then it is easy to see that this results in
a well-structured transition system
(WSTS)~\cite{DBLP:journals/tcs/FinkelS01,DBLP:conf/lics/AbdullaCJT96}. This
also implies, e.g.\ that termination is decidable in this general setting.

	\newcommand{\completeAbbr}{complete}
	\begin{table}[t]
		\begin{center}
			\begin{tabular}{lccc} \toprule
				\makecell{Dimension \\ \& encoding}& Monus Reachability & Monus zero reachability & Monus coverability \\\midrule
				1-dim, unary   & $\NL$-\completeAbbr        & $\NL$-\completeAbbr             & $\NL$-\completeAbbr \\
				1-dim, binary   & $\NP$-\completeAbbr        & in $\NC^2$               & in $\NC^2$ \\
				2-dim, unary   & in $\PSPACE$        &               $\NL$-\completeAbbr & $\NL$-\completeAbbr\\
				2-dim, binary   & $\PSPACE$-\completeAbbr        & $\PSPACE$-\completeAbbr               & $\NP$-\completeAbbr \\
				arbitrary   & Ack-\completeAbbr        & $\EXPSPACE$-\completeAbbr               & $\NP$-\completeAbbr \\\bottomrule
			\end{tabular}
		\end{center}
		\caption{Complexity results shown in this work.}\label{table-results}
	\end{table}
	
	\renewcommand{\completeAbbr}{complete}
	\newcommand{\completeAbbrShort}{compl.}
	\begin{table}[t]
		\begin{center}
			\begin{tabular}{lccc} \toprule
				\makecell{Dimension \\ \& encoding} & Reachability & Zero reachability & Coverability \\\midrule
				1-dim, unary   & $\NL$-\completeAbbr~\cite{ValiantP75}        & $\NL$-\completeAbbr~\cite{ValiantP75}             & $\NL$-\completeAbbr~\cite{ValiantP75} \\
				1-dim, binary   & $\NP$-\completeAbbr~\cite{DBLP:conf/concur/HaaseKOW09}        & $\NP$-\completeAbbr~\cite{DBLP:conf/concur/HaaseKOW09}               & in $\NC^2$~\cite{DBLP:conf/concur/AlmagorCPS020} \\
				2-dim, unary   & $\NL$-\completeAbbr~\cite{DBLP:journals/jacm/BlondinEFGHLMT21}        &     $\NL$-\completeAbbr~\cite{DBLP:journals/jacm/BlondinEFGHLMT21} & $\NL$-\completeAbbr~\cite{rosier1986multiparameter}\\
				2-dim, binary   & $\PSPACE$-\completeAbbr~\cite{DBLP:journals/jacm/BlondinEFGHLMT21}        & $\PSPACE$-\completeAbbr~\cite{DBLP:journals/jacm/BlondinEFGHLMT21}               & $\PSPACE$-\completeAbbr~\cite{DBLP:journals/jacm/BlondinEFGHLMT21,rosier1986multiparameter,FearnleyJurdzinski13} \\
				arbitrary   & Ack-\completeAbbrShort~\cite{DBLP:conf/lics/LerouxS19,DBLP:conf/focs/Leroux21,DBLP:conf/focs/CzerwinskiO21}        & Ack-\completeAbbrShort~\cite{DBLP:conf/lics/LerouxS19,DBLP:conf/focs/Leroux21,DBLP:conf/focs/CzerwinskiO21}               & $\EXPSPACE$-\completeAbbrShort~\cite{lipton1976reachability,Rackoff78} \\\bottomrule
			\end{tabular}
			
		\end{center}
		\caption{Known complexities for classical VASS semantics, for comparison.}\label{table-classical}
	\end{table}
	
	\section{Arbitrary dimension}\label{sec:arbitrary_dim}
	
	In this section, we prove the complexity results concerning VASS with arbitrary dimension. This will include the characterizations of monus reachability, monus zero reachability, and monus coverability in terms of classical and $\Z$-semantics. We begin with some terminology.
	
	\subparagraph*{Paths.}
	A sequence of transitions $(p_1,\zz_1,q_1),\ldots,(p_{k},\zz_{k},q_{k})$ is valid iff $q_i = p_{i+1}$ for every $1 \le i < k-1$. Furthermore, we say that it is valid from a given configuration $(p, \vv)$ if $p = p_0$. We call a valid sequence of transitions a \emph{path}. 

	Given two paths $\rho_1$ and $\rho_2$ if the last state of $\rho_1$ is equal to the first state of $\rho_2$ then by $\rho = \rho_1 \rho_2$ we denote the path defined as the sequence $\rho_1$ followed by the sequence $\rho_2$.
	Similarly, we use this notation with more paths, \eg $\rho = \rho_1 \rho_2 \ldots \rho_k$ means that the path $\rho$ is composed from $k$ paths: $\rho_1, \ldots \rho_k$. 

	Fix a path $\rho = (p_0,\zz_0,p_1),\ldots,(p_{k-1},\zz_{k-1},p_k)$.
	We say that $\zz = \sum_{i=0}^{k-1}\zz_i$ is the effect of the path $\rho$.
	Notice that while for classical and $\Z$-semantics the effect of a path can be computed by subtracting the vectors in the last and first configurations, this is not necessarily true for monus semantics. In \cref{fig:intro_semantics} consider the path $\rho = t,t,t$. The effect is $(-3,6)$. In the $\Z$-semantics $(2,0) \zsteps (-1,6)$ and the difference $(-1,6) - (2,0)$ is precisely the effect of $\rho$. In the monus semantics it is not the case as $(2,0) \weaksteps (0,6)$. This is because a run in monus semantics can lose some decrements, unlike in classical and $\Z$-semantics.

	\begin{remark}\label{rem:vass}
		Observe that every classical and $\Z$ run defines a unique path from the initial configuration. For monus semantics uniqueness is not guaranteed as it is possible that a run induces more than one path. Indeed, suppose $p(2,0) \weakstep q(1,0)$. This could be realised by any transition of the form $(p,(-1,z),q)$, where $z \le 0$. 
		Conversely, a path induces a unique run for $\Z$ and $\weak$ semantics. 
		Formally, consider a path $(p_0,\zz_1,p_1),\ldots,(p_{k-1},\zz_k,p_{k})$ from a configuration $s(\vv)$. Then, in the $\Z$ and $\weak$ semantics there exists a unique corresponding run. In the classical semantics a path might be blocked if a counter drops below zero (see \eg \cref{fig:intro_semantics}).
		We write $p_0(\vv_0) \xsteps{\rho} p_k(\vv_k)$, $p_0(\vv_0) \xzsteps{\rho} p_k(\vv_k)$ and $p_0(\vv_0) \xweaksteps{\rho} p_k(\vv_k)$ if $p_0(\vv_0), \ldots, p_k(\vv_k)$ is a run in classical, integer and monus semantics, respectively. Recall that for classical and $\Z$-semantics $\vv_{i+1} - \vv_i = \zz_i$, and for monus semantics $\vv_{i+1} = \max(\vv_i+\zz_i, \bzero)$.
	\end{remark}
	
	Consider a run $R = p_0(\vv_0), \ldots, p_k(\vv_k)$ (in any semantics). We say that the counter $j \in \{1,\cdots, d\}$ hits $0$  
	iff $\vv_i[j] = 0$ for some $1\le i \le k$. Similarly, we say that the counter $j \in \{1,\cdots, d\}$ goes negative in $R$ iff $\vv_i[j] < 0$ for some $0\le i \le k$ (this can happen only in the $\Z$-semantics).
	
	Let $\rho = (p_0, \zz_0, p_1) \ldots (p_{k-1}, \zz_{k-1}, p_k)$  be a path such that $R$ is the unique run corresponding to $\rho$ from the initial configuration $p_0(\vv_0)$. We say that $(\rho, R)$ or $p_0(\vv_0) \xweaksteps{\rho} p_k(\vv_k)$ is lossy for the counter $j \in \{1,\cdots, d\}$ iff $\vv_i[j] - \vv_{i-1}[j] \ne \zz_{i-1}[j]$ for some $1\le i \le k$ (a lossy run can happen only in the monus semantics).
	\begin{remark} \label{rem:vass-approxim}
		Integer and monus semantics are overapproximations of the classical semantics. That is, $s(\vv) \xsteps{\rho} t(\ww)$ implies $s(\vv) \xzsteps{\rho} t(\ww)$ and $s(\vv) \xweaksteps{\rho} t(\ww)$. The converse is not always the case (see \cref{fig:intro_semantics}). Moreover, $s(\vv) \xweaksteps{\rho} t(\ww)$ implies $s(\vv) \xsteps{\rho} t(\ww)$ if $s(\vv) \xweaksteps{\rho} t(\ww)$ is not lossy. Notice that if in $s(\vv) \xweaksteps{\rho} t(\ww)$, none of the counters $j \in \{1,\ldots,d\}$ hits $0$ then it is not a lossy run. Similarly,  $s(\vv) \xzsteps{\rho} t(\ww)$ implies $s(\vv) \xsteps{\rho} t(\ww)$ if, in the former run, none of the counters $j \in \{1,\ldots,d\}$ goes negative.
	\end{remark}

	\subparagraph{Characterizing Monus Reachability.}
	Our first goal is to characterize the reachability problem for the monus semantics in terms of the classical semantics.
	We start with some propositions that relate monus runs to $\Z$ runs and classical runs.
	Let $\rho$ be a path and $s_0(\vv_0)$ a configuration. Let $s_0(\vv_0) \ldots s_k(\vv_k)$ be the unique $\Z$ run defined by  $\rho$ and $s_0(\vv_0)$. We define the vector $\mm = \mi(\rho, s_0, \vv_0)$ by $\mm[i] = \min(\min_{j = 0}^{k} \vv_j[i], 0)$. Intuitively, it is the vector of minimal values in the $\Z$ run, but note that $\mm \le \bzero$.
	
	For the next two propositions we fix a configuration $s_0(\vv_0) \in Q \times \N^d$, a path $\rho = (s_0, \zz_0, s_1) \ldots (s_{k-1}, \zz_{k-1}, s_k)$, and $\mm = \mi(\rho, s_0, \vv_0)$. 
	
	\begin{proposition}\label{prop:weak-integer}
		Consider the unique runs induced by $\rho$ from $s_0(\vv_0)$ in $\Z$-semantics
		\[
		s_0(\vv_0), \ldots, s_{k-1}(\vv_{k-1}), s_k(\vv_k),
		\]
		and in monus semantics
		\[
		s_0(\vv'_0), \ldots, s_{k-1}(\vv'_{k-1}), s_k(\vv'_k).
		\]
		where $\vv'_0 = \vv_0$. Then $\vv'_k = \vv_k - \mm$.
	\end{proposition}
	\begin{proof}[Proof (sketch)]
		We analyse the behavior of every counter $j$. Recall that the $\Z$ run and the monus run have the same value in the counter $j$ until the first time the value of $j$ becomes negative in the $\Z$ run. We denote this as $\vv_i[j] = -u$. Note that $\vv'_i[j] = 0$. Hence, $\vv_i[j] - \vv'_i[j] = -u$. It is not hard to see that every time the value of the counter $j$ reaches a new minimum in the $\Z$-semantics, the difference $\vv'_i[j] - \vv_i[j]$ will be equal to it. We prove this formally by induction on $k$. Refer to \cref{app:weak-integer} for the formal proof.  
	\end{proof}
	
	\begin{remark}
		\label{rem:zvass}
		Let $\zz \in \Z^d$.
		A sequence of configurations $s_0(\vv_0) \ldots s_k(\vv_k)$ is a run in $\Z$-semantics corresponding to a path $\rho$ iff $s_0 (\vv_0 - \zz) \ldots s_k(\vv_k - \zz)$ is a run in $\Z$-semantics on the same path $\rho$.
	\end{remark}
	
	\begin{proposition}
		\label{prop:weak-class}
		Consider the following unique run corresponding to the path $\rho$ from $s_0(\vv_0)$ in the monus semantics
		\[
		s_0(\vv_0), \ldots, s_{k-1}(\vv_{k-1}), s_k(\vv_k).
		\] Then the following run, induced by $\rho$, exists in the classical semantics
		\[
		s_0(\vv'_0), \ldots, s_{k-1}(\vv'_{k-1}), s_k(\vv'_k).
		\]
		where $\vv'_0 = \vv_0 - \mm$ and $\vv'_k = \vv_k$.
	\end{proposition}
	\begin{proof}
		This essentially follows from the definition of $\mm$ and \cref{rem:zvass}. One just needs to observe that the $\Z$ run with configurations shifted by the vector $-\mm$ does not go below zero, hence it is a classical run.
		See \cref{app:weak-class} for the formal proof.
	\end{proof}
	
	We now characterize monus reachability in terms of classical reachability.
	\begin{proposition}\label{charact-reachability}
		Let $\cV=(d,Q,\Delta)$ be a VASS, let $s(\vv)$ and $t(\ww)$ be
		configurations of $\cV$, and let $\rho$ be a path of $\cV$. Then, $s(\vv) \xweaksteps{\rho} t(\ww)$ if and only
		if there is a subset $Z\subseteq\{1,\ldots, d\}$ and a vector $\vv' \ge \vv$ such that
		\begin{enumerate}
			\item\label{charact-reachability:steps} $s(\vv') \xsteps{\rho} t(\ww)$,
			\item\label{charact-reachability:add-minimum} For every $z\in Z$, the coordinate $z$ hits $0$ in $s(\vv') \xsteps{\rho} t(\ww)$,
			\item\label{charact-reachability:initial} For every $j\in\{1,\ldots,d\}\setminus Z$, we have $\vv'[j]=\vv[j]$.
		\end{enumerate}
	\end{proposition}
	\begin{proof}
		$(\implies)$
		Let $\mm = \mi(\rho, s, \vv)$.
		This direction is implied by \cref{prop:weak-class} along with the following argument. Every counter $j \in \{1,\ldots, d\}$ hits $0$ in $s(\vv) \xweaksteps{\rho} t(\ww)$ if and only if it hits $0$ in $s(\vv - \mm) \xsteps{\rho} t(\ww)$. Moreover, if $j$ does not hit $0$ in  $s(\vv) \xweaksteps{\rho} t(\ww)$ then $\mm[j] = 0$.  
	
		$(\impliedby)$
		Let $\vv' \ge \vv$ be a vector as in the statement and let $s(\vv') \xsteps{\rho} t(\ww)$. We define $Z \subseteq \{1 \ldots d\}$ such that $i \in Z$ if it hits $0$.
		Moreover, let $s(\vv) \xweaksteps{\rho} t(\ww'')$.
		It suffices to show that $\ww = \ww''$.
		We write $s(\vv') = p_0(\vv'_0) \ldots p_k(\vv'_k) = t(\ww)$ and $s(\vv) = p_0(\vv_0) \ldots p_k(\vv_k) = t(\ww'')$ for the corresponding runs in the classical and monus semantics, respectively. 
		Note that $\vv'\ge \vv$ implies $\vv'_i \ge \vv_i$ for all $0 \le i \le k$. By definition of $\vv'$ it suffices to consider counters $j$ that hit zero, \ie $\vv'_i[j] = 0$ for some $0 \le i \le k$. Since $\vv'_i \ge \vv_i$ we get $\vv_i'[j] = 0 = \vv_i[j]$. Hence, from $i$ onward both runs agree on the value in counter $j$. Thus $\ww = \ww''$.
    
    See \cref{app:charact-reachability} for an extended version of this proof.
		\end{proof}

	\subparagraph{The reachability problem.}
	We begin with the Ackermann-completeness proof.
	\begin{theorem}\label{reachability-ackermann}
		Reachability in monus semantics is Ackermann-complete.
	\end{theorem}
	For the upper bound we show how to reduce reachability in monus semantics to reachability in classical semantics.
	Let $\cV=(d,Q,\Delta)$, $s(\vv)$, and $t(\ww)$ be the input of the reachability problem in monus semantics.
	We rely on \cref{charact-reachability}. Intuitively, we have to guess a subset $Z\subseteq\{1,\ldots,d\}$ and a permutation $\sigma\colon[1,k]\to Z$ (where $k=|Z|$). Then we check whether there exists a run as described in \cref{charact-reachability} with $z_i=\sigma(i)$ for $i\in[1,k]$. To detect the latter run, we construct the VASS $\cV_\sigma = (d+k,Q',T')$ as follows. It simulates $\cV$, but it has $k$ extra counters to freeze the values of the counter in $Z$ at the points where the coordinates $\sigma(k), \ldots, \sigma(1)$ hit $0$ as mentioned in \cref{charact-reachability}.
	
	To remember which counters have already been frozen the set of control states is $Q' = \{q_i\mid q\in Q,~i\in[0,k]\}$. Intuitively, the index $i\in[0,k]$ stores the information how many counters are frozen. The index $i$ can only increment. Note that guessing the permutation $\sigma$ allows us to assume that we know the order in which the counters are frozen.
	
	Since we deal with vectors in dimension $d$ and $d+k$ we introduce some helpful notation.
	We write $\ee_j \in \Z^d$ for the unity vector with $\ee_j[j] = 1$ and with $0$ on other coordinates.
	Given a vector $\zz \in \Z^d$ we define $\copyv{\zz}\in\Z^{d+k}$ as $\copyv{\zz}[j] = \zz[j]$ for $1 \le j \le d$ and $\copyv{\zz}[j] = \zz[\sigma(j-d)]$ for $d < j \le d+k$.
	Intuitively, it simply copies the behaviors of the corresponding counters.
	We generalise this notation to allow to also remove the effect on some coordinates
	(\ie ``freeze'' them).
	Given $\zz \in \Z^{d}$ and $0 \le i \le k$ we define $\copyi{\zz} \in \Z^{d+k}$ as $\copyi{\zz}[j] = \copyv{\zz}[j]$ for $1 \le j \le d+ k-i$ and $\copyi{\zz}[j] = 0$ for $d+ k-i < j \le d+k$. In particular $\mathsf{copy}_0(\zz) = \mathsf{copy}(\zz)$ and $\copyi{\zz}$ is $0$ in the last $i$ counters.
	
	It remains to define the set of transitions $T'$. 
	In the beginning there are transitions in $T'$ that can arbitrarily increment each counter that belongs to $Z$ and its extra copy: $(s_0,\copyv{\ee_j},s_0) \in T'$ for every $j \in Z$. Moreover, the counter in the control state can spontaneously be incremented: $(p_i,\bzero,p_{i+1})$ for every $p\in Q$ and $0\le i < k$.
	For every transition $(p,\zz,q) \in T$ and $0 \le i \le k$ we define $(p_i,\copyi{\zz},q_i) \in T'$.
	
	The following claim is straightforward by \cref{charact-reachability}:
	\begin{claim}
		We have $s(\vv) \mathrel{{\weaksteps}{}_{\cV}} t(\ww)$ if and only if there exists a subset $Z\subseteq\{1,\ldots,d\}$ and bijection $\sigma\colon [1,k]\to Z$ such that $s_0(\mathsf{copy}_0(\vv))\steps_{\cV_\sigma} t_k(\mathsf{copy}_k(\ww))$.
	\end{claim}
	This implies that we can decide monus reachability by guessing a subset $Z\subseteq[1,d]$, guessing a bijection $\sigma\colon[1,k]\to Z$, and deciding reachability in $\cV_\sigma$. This yields the upper bound.
	
	For the lower bound we reduce classical reachability to monus reachability.
	Let $\cV=(d,Q,\Delta)$, $s(\bzero)$ and $t(\bzero)$ be the input of the reachability problem in classical semantics (without loss of generality the input vectors can be $\bzero$).
	We construct the VASS $\cV' = (d+2,Q',T')$ as follows. The states are $Q' = Q \cup \set{t'}$, where $t'$ is a fresh copy of $t$.
	
	Again to deal with vectors in different dimension we introduce the following notation. Given $\zz \in \Z^d$ we write $\Delta(\zz) \in \Z$ for $\Delta(\zz) = \sum_{j = 1}^d \zz[j]$, \ie the sum of all components.
	Based on this we define $\extend{\zz} \in \Z^{d+2}$ as:
	$\extend{\zz} = (z,\Delta(z),0)$ if $\Delta(\zz) \ge 0$, and $\extend{\zz} = (z,0,-\Delta(z))$ otherwise.
	
	We define $T'$ as follows. For every $(p, \zz,q) \in T$: $(p,\extend{\zz},q) \in T'$.
	Thus, in the 
	$(d+1)$-th counter, we collect the sum of all non-negative entry sums of the added
	vectors. Analogously, in the $(d+2)$-th counter, we collect the sum of all negative entry sums (with a flipped sign). 
	We also add the transition $(t,\bzero, t') \in T'$, and a ``count down'' loop: $(t'(\bzero,-1,-1), t')$, where $(\bzero,-1,-1)$ is $0$ in the first $d$ components and $-1$ otherwise. The following claim completes the proof of Ackermann-hardness.
	\begin{claim}
		We have $s(\bzero,1,1)\weaksteps t'(\bzero,1,1)$ in $\cV'$ if and only if
		$s(\bzero)\steps t(\bzero)$ in $\cV$.
	\end{claim}
	
	\begin{claimproof}
		($\impliedby$)
		This is obvious, because every run in classical semantics yields a run in
		monus semantics between the same configurations.
		
		($\implies $)
		Suppose there is a monus run from $s(\bzero,1,1)$ to
		$t'(\bzero,1,1)$. Then for some $m\in\N$, there is a transition sequence $\rho$
		leading in monus semantics from $s(\bzero,1,1)$ to $t(\bzero,m,m)$. Now let us
		execute $\rho$ in $\Z$-semantics. This execution will arrive at some
		configuration $t(\vv,m,m)$ (note that the last two counters are never decreased, except for the final loop).  We shall prove that (i)~$\vv=\bzero$ and (ii)~this
		execution never drops below zero. First, according to \cref{prop:weak-integer},
		the resulting counter values in monus semantics are always at least the values
		from $\Z$-semantics.  This implies $\vv\le\bzero$.  Next observe that since the
		right-most components have the same value $m$, the total sum of all entry sums
		of added vectors (in the first $d$ entries) must be zero. Thus, $\Delta(\vv)=0$.
		Together with $\vv\le\bzero$, this implies $\vv=\bzero$, which shows (i).
		Second, if the execution in $\Z$-semantics ever drops below zero in some
		counter $i$, then by \cref{prop:weak-integer} and the fact that in
		$\Z$-semantics we reach $\vv=\bzero$, this would imply that $\rho$ in monus
		semantics ends up in a strictly positive value in counter $i$, which is not
		true. This shows (ii).  Hence, we have shown that the run in $\Z$-semantics is
		actually a run in classical VASS semantics. Therefore, $s(\bzero)\steps
		t(\bzero)$ in $\cV$.
	\end{claimproof}

	\subparagraph{Characterizing zero-reachability.} 
	Monus zero-reachability has a simple characterization in terms of classical coverability.
	Here, $\rev{\cV}$ is obtained by reversing all transitions in $\cV$ and their effects. Formally, there is a transition $(p, \zz, q)$ in $\rev{\cV}$ iff there is a transition $(q, -\zz, p)$ in $\cV$.
	\begin{lemma}\label{charact-zero-reachability}
		For any $\vv$, we have $s(\vv) \mathrel{{\weaksteps}{}_{\cV}}
		t(\mathbf{0})$ iff 
		$t(\mathbf{0})\steps_{\rev{\cV}}s(\vv')$ for some $\vv' \ge \vv$.
	\end{lemma}
	\begin{proof}
		By \cref{charact-reachability},  $s(\vv) \weaksteps t(\mathbf{0})$
		yields a $\vv' \ge \vv$ with $s(\vv') \steps t(\mathbf{0})$.
		Conversely, if $s(\vv')\steps t(\bzero)$, then we can pick $Z=[1,d]$
		in \cref{charact-reachability} to obtain $s(\vv)\weaksteps
		t(\bzero)$.
	\end{proof}
	This together with the known complexity of classical coverability \cite{lipton1976reachability,Rackoff78} immediately implies:
	
	\begin{proposition}\label{weak-zero-reachability-expspace}
		The monus zero-reachability problem is $\EXPSPACE$-complete.
	\end{proposition}
	
	\subparagraph{Characterizing coverability.} Our third characterization describes coverability in monus semantics in terms of reachability in $\Z$-semantics:
	\begin{proposition}\label{charact-coverability}
		Let $\cV=(d,Q,\Delta)$ be a VASS and let $s(\vv)$ and $t(\ww)$ be
		configurations. Then $s(\vv) \weaksteps t(\ww'')$ for some $\ww'' \ge \ww$ if and only
		if there is a permutation $\sigma$ of $\{1,\ldots,d\}$ and $\Z$-configurations
		$p_d(\vv_d),\ldots,p_1(\vv_1)$, $t(\ww')$ so that
		\begin{enumerate}
			\item $s(\vv)\zsteps p_d(\vv_d)\zsteps p_{d-1}(\vv_{d-1})\zsteps\cdots \zsteps p_1(\vv_1)\zsteps t(\ww')$,
			\item for each $j\in \{1,\ldots,d\}$, we have $\ww'[j]+|\min(\vv_{\sigma^{-1}(j)}[j], 0)|\ge \ww[j]$.
		\end{enumerate}
	\end{proposition}
	\begin{proof}
		($\implies$)
		Let $\rho$ be any path such that $s(\vv) \xweaksteps{\rho} t(\ww'')$ and $\ww'' \ge \ww$. Then, by \cref{prop:weak-integer} $s(\vv) \xzsteps{\rho} t(\ww''+\mm)$, where $\mm$ is the vector of minimum values in the $\Z$ run. The required permutation $\sigma$ represents the order $\sigma(d), \ldots, \sigma(1)$ in which these coordinates reach their corresponding minimum values.
		Hence, $s(\vv) \xzsteps{\rho} t(\ww''+\mm)$ is the same as $s(\vv)\zsteps p_d(\vv_d)\zsteps p_{d-1}(\vv_{d-1})\zsteps\cdots \zsteps p_1(\vv_1)\zsteps t(\ww')$, such that $\vv_d[\sigma(d)] = \mm[\sigma(d)], \ldots, \vv_1[\sigma(1)] = \mm[\sigma(1)]$, and $\ww''[j] = \ww'[j]-\mm[j] = \ww'[j]+|\mm[j]| = \ww'[j]+|\min(\vv_{\sigma^{-1}(j)}[j], 0)|$ for all $1\le j \le d$. As $\ww'' \ge \ww$,  $\ww'[j]+|\min(\vv_{\sigma^{-1}(j)}[j], 0)| \ge \ww[j]$ for all $1 \le j \le d$.  
		
		($\impliedby$)
		This is a direct consequence of \cref{prop:weak-integer}. It implies that given any permutation
		$\sigma$ on $\{1,\ldots,d\}$ and any run  
		$s(\vv)\zsteps p_d(\vv_d)\zsteps p_{d-1}(\vv_{d-1})\zsteps\cdots \zsteps p_1(\vv_1)\zsteps t(\ww')$ such that $\ww'[j] -\min(\vv_{\sigma^{-1}(j)} [j], 0) \ge \ww[j]$, 
		there is a run from configuration  $s(\vv)$ and reaching a configuration  $t(\ww'')$ where 
		$\ww''[j] = \ww'[j] - \mm[j] \ge \ww'[j] -\min(\vv_{\sigma^{-1}(j)}[j], 0) \ge \ww[j]$ for all $1\le j \le d$.  
	\end{proof}

	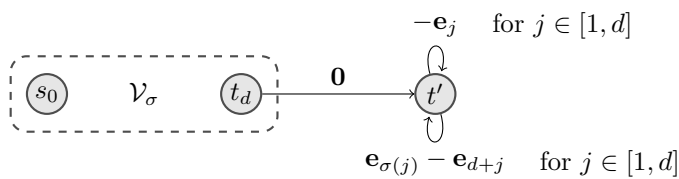
\begin{figure}
		\begin{center}
			\begin{tikzpicture}
				\node[state] (s) {$s_0$};
				\node[state,right=2cm of s] (t) {$t_d$};
				\node[state,right=2cm of t] (t') {$t'$};
				\path[->]
				(t) edge node[above] {$\bzero$} (t')
				(t') edge[loop above] node (jd) {$-\ee_j$} (t')
				(t') edge[loop below] node (jz) {$\ee_{\sigma(j)}-\ee_{d+j}$} (t')
				;
				\node[right=0.2cm of jd] {for $j\in [1,d]$};
				\node[right=0.2cm of jz] {for $j\in [1,d]$};
				
				\node (vsigma) at ($(s.center)!0.5!(t.center)$) {$\cV_\sigma$};
				\draw[dashed,rounded corners=6pt,color=black!70,thick] ($(s.west)+(-0.2,-0.5)$) rectangle ( $(t.east)+(0.2,0.5)$ );
				
			\end{tikzpicture}
			\caption{Construction of $\cV'_\sigma$ in reduction from monus coverability to reachability in $\Z$-semantics.}\label{coverability-to-integer}
		\end{center}
	\end{figure}
	We conclude the following.
	\begin{proposition}\label{weak-coverability}
		Monus coverability is $\NP$-complete.  
	\end{proposition}
	
	\begin{proof}
		First we show $\NP$-hardness. In
		\cite[Prop.~5.11]{DBLP:journals/corr/Kopczynski15}, it is shown that it is
		$\NP$-hard to decide whether a regular language over some alphabet $\Sigma$,
		given as an NFA, contains a word in which every letter appears exactly once.
		Given such an NFA $\cA$ over $\Sigma=\{a_1,\ldots,a_d\}$, we construct a
		$d$-VASS $\cV$. The VASS $\cV$ simulates $\cA$ such that when $\cA$ reads
		$a_i$, $\cV$ increments counter $i$. Moreover, $\cV$ maintains a number
		$k\in\{0,\ldots,d\}$ in its state, which always holds the number of letters
		read so far. Thus, $\cV$ has states $q_k$, where $q$ is  a state of $\cA$ and
		$k\in\{1,\ldots,d\}$.  Moreover, let $s$ and $t$ be the initial and final state
		of $\cA$, respectively. Then in $\cV$, one can cover $t_d(1,\ldots,1)$ from
		$s_0(\bzero)$ in monus semantics if and only if $\cA$ accepts some word as above.
		
		We turn to the $\NP$ upper bound. Suppose we are given a
		$d$-VASS $\cV=(d,Q,\Delta)$ and configurations $s(\uu),t(\vv)$. We employ
		\cref{charact-coverability}. First non-deterministically guess a
		permutation $\sigma$ of $[1,d]$. We now construct a $2d$-VASS $\cV'_\sigma$ and two
		configurations $c'_1,c'_2$ such that in $\cV'_\sigma$, we have $c'_1\zsteps
		c'_2$ if and only if there is a run as in \cref{charact-coverability} with this
		$\sigma$. Since reachability in $\Z$-semantics is
		$\NP$-complete~\cite{DBLP:conf/rp/HaaseH14}, this yields the upper bound. 
		
		Our VASS $\cV'_\sigma$ is a slight extension of the VASS $\cV_\sigma$ from
		\cref{reachability-ackermann}, see \cref{coverability-to-integer}. Recall that
		for a permutation $\sigma\colon[1,k]\to Z$, $\cV_\sigma$ keeps $k$ extra
		counters that freeze the values of the counters in $Z$, in the order
		$\sigma(k),\sigma(k-1),\ldots,\sigma(1)$. We use this construction, but for our
		permutation $\sigma$ of $[1,d]$. Thus, $\cV_\sigma$ simulates a run of $\cV$
		and then freezes the counters $\sigma(d),\ldots,\sigma(1)$ in the extra $d$
		counters, in this order.  The steps that freeze counters define the vectors
		$\vv_d$, \ldots, $\vv_1$ in \cref{charact-coverability}. Note that for each
		$\vv_{i}$, only $\vv_{i}[\sigma(i)]$ is important.
		
		To verify the second condition in \cref{charact-coverability}, we introduce an extra state $t'$ and extra transitions as depicted in \cref{coverability-to-integer}.
		After executing $\cV_\sigma$, $\cV'_\sigma$ then has two
		types of loops: One to move tokens from the counters $d+j$ to counters
		$\sigma(j)$ (for each $j\in[1,d]$), and one to reduce tokens in counters $1,\ldots,d$.
		Thus there exists $\sigma$ such that $s_0(\mathsf{copy}_0(\uu)) \zsteps t'(\mathsf{copy}_d(\vv))$ in $\cV'_\sigma$ if and
		only if $s(\uu) \weaksteps t(\vv'')$ for some $\vv'' \ge \vv$ in $\cV$. This
		proves the $\NP$ upper bound. 
	\end{proof}

	\section{Two-dimensional VASS}\label{sec:two_dim}
	In this section we prove the results of \cref{table-results} related to $2$-VASS,
	both for unary and binary encoding.
	Note that for all three considered problems, reachability, zero reachability, and coverability,
	we always have an $\NL$ lower bound, inherited from state reachability in finite automata.
	The latter is well-known to be $\NL$-hard, and a VASS without counters (in all considered semantics) is a finite state automaton.

	When dealing with binary/unary updates one needs to be careful with the input size. In all problems suppose a VASS $\cV = (d,Q,T)$ is in the input. If we are interested in the unary encoding its size is defined as $d + |Q| + \sum_{(p,\zz,q) \in T}\lVert \zz \rVert$, where $\lVert \zz \rVert$ is the absolute value of the maximal coordinate in $\zz$. In the binary encoding one needs to change $\lVert \zz \rVert$ to $\lceil\log (\lVert \zz \rVert + 1)\rceil$. From this point onwards, we use the term \textit{succinct} VASS for VASS where updates are encoded in binary. 
	
	We consider each of the three problems separately.

	\subparagraph{Reachability}
	Here we only prove the $\PSPACE$ upper bound for monus reachability in binary encoded $2$-VASS,
	which implies the same upper bound for unary encoding.
	The $\PSPACE$ lower bound for binary encoding is inherited from zero reachability, see \cref{charact-zero-reach-2-VASS} below. 
	\begin{proposition}\label{weak-reach-2-VASS-upper}
		In succinct 2-VASS, reachability with monus semantics is in $\PSPACE$.
	\end{proposition}
	According to \cref{charact-reachability}, reachability with
	monus semantics is equivalent to existence of a run under classical semantics,
	where said run is subject to some additional constraints.
	Recall that \emph{Presburger arithmetic} is the first-order theory of $(\N,+,<,0,1)$. 
	We observe that all the additional constraints of \cref{charact-reachability}
	can be expressed by quantifier-free Presburger formulas.
	This leads us to the so-called \emph{constrained runs problem} for succinct $2$-VASS,
	which was recently shown to be in $\PSPACE$~\cite{BaumannMeyerZetzsche2022a}, following
	the fact that classical reachability itself is $\PSPACE$-complete for succinct $2$-VASS~\cite{DBLP:journals/jacm/BlondinEFGHLMT21}.
	
	Formally, the \emph{constrained runs problem} for succinct $2$-VASS is the following:
	\begin{description}
		\item[Given] A succinct $2$-VASS $\cV$, a
		number $m\in\N$, states $q_1,\ldots,q_m$ in $\cV$,
		a quantifier-free Presburger formula
		$\psi(x_1,y_1,\ldots,x_m,y_m)$,
		and numbers $s,t\in[1,m]$ with $s\le t$.
		\item[Question] Does there exist a run
		$q_0(0,0)$ $\steps$ $q_1(x_1,y_1)$ $\steps$ $\cdots$ $\steps$ $q_m(x_m,y_m)$
		that visits a final state between $q_s(x_s,y_s)$ and
		$q_t(x_t,y_t)$ and satisfies
		$\psi(x_1,y_1,\ldots,x_m,y_m)$?
	\end{description}
	\begin{lemma}[{\cite[Prop.~6.5]{BaumannMeyerZetzsche2022a}}] \label{constrained-runs-PSPACE}
		The constrained runs problem for succinct $2$-VASS is in $\PSPACE$.
	\end{lemma}
	
	We can now prove \cref{weak-reach-2-VASS-upper} by reducing to the constrained runs problem:
	Let $\cV$ be a $2$-VASS with configurations $s(\vv)$ and $t(\ww)$.
	According to \cref{charact-reachability}, existence of a run $s(\vv) \weaksteps t(\ww)$ is equivalent to
	existence of states $p_1, p_2$ and a set $Z \subseteq [1,2]$ such that a run
	$s(\vv') \steps t(\ww)$ with $\vv' \geq \vv$ that is subject to additional requirements enforced by conditions (2) and (3) of the \cref{charact-reachability}. 
	Our $\PSPACE$ algorithm enumerates all possibilities of $p_1,p_2$ and $Z$,
	constructing an instance of the constrained run problem each time,
	and checking for a constrained run in $\PSPACE$ using \cref{constrained-runs-PSPACE}.
	If such a run exists in at least one of the instances, the algorithm accepts, otherwise it rejects.
	To construct each instance the algorithm first modifies $\cV$ to ensure that a starting configuration
	$s(\vv')$ is reachable for any $\vv' \geq \vv$.
	To this end a new initial state $q_0$ is added, with two loops that increment one of the counters each,
	and a transition that goes to $s$ by adding $\vv$.
	Then the additional requirements of \cref{charact-reachability} are encoded in
	quantifier-free Presburger arithmetic, as required by the constrained run problem.
	Clearly the constructed algorithm runs in $\PSPACE$ and decides $s(\vv) \weaksteps t(\ww)$.
	For more details refer to \cref{proof-reachability-2-VASS-upper}.

	\subparagraph{Zero reachability}
	
	\begin{proposition}\label{charact-zero-reach-2-VASS}
		Monus zero reachability in $2$-VASS is $\PSPACE$-complete under binary encoding and $\NL$-complete under unary encoding.
	\end{proposition}
	\begin{proof}
		This is a simple consequence of monus zero reachability being interreducible with classical coverability: 
Classical coverability in $2$-VASS under binary encoding is $\PSPACE$-complete under binary encoding (in~\cite[Corollary 3.3]{DBLP:journals/jacm/BlondinEFGHLMT21}, this is deduced from~\cite[p.~108]{rosier1986multiparameter} and~\cite[Corollary 10]{DBLP:conf/icalp/FearnleyJ13} and $\NL$-complete under unary encoding~\cite[p.~108]{rosier1986multiparameter}.
		
		Let $\cV$ be a $2$-VASS with configurations $s(\vv)$ and $t(\bzero)$.
		Then according to \cref{charact-reachability}, we know that $t(\bzero)$ is monus
		reachable from $s(\vv)$ if and only if in $\rev{\cV}$ the
		configuration $s(\vv)$ is coverable from $t(\bzero)$ with classical semantics.
		On the other hand, given configurations $s(\vv)$ and $t(\ww)$ of a $2$-VASS $\cV$,
		we add a new state $s'$ and transition $(s',\vv,s)$ to construct the $2$-VASS $\cV'$.
		Then classical coverability of $t(\ww)$ from $s(\vv)$ in $\cV$ is equivalent
		to the same from $s'(\bzero)$ in $\cV'$.
		Now applying \cref{charact-reachability} in reverse, the latter is further equivalent to
		monus reachability of $s'(\bzero)$ from $t(\ww)$ in $\rev{\cV'}$.
	\end{proof}
	
	\subparagraph{Coverability}
	By \cref{weak-coverability}, monus coverability is in $\NP$ in arbitrary dimension.
	Thus, it remains to show the $\NP$ lower bound.
	
	\begin{figure}
		\begin{center}
			\begin{tikzpicture}[initial text={}]
				\node[state,initial] (q0) {$s$};
				\node[state,right=1.6cm of q0] (q1) {};
				\node[state,right=1.2cm of q1] (q2) {};
				\node[state,right=1.2cm of q2] (q3) {};
				\node[state,right=1.2cm of q3] (q4) {};
				\node[state,right=1.2cm of q4,accepting] (q5) {$t$};
				\node at ($(q2.center)!0.5!(q3.center)$) (d) {$\cdots$};
				\path[->] 
				(q0) edge node[above] {$(1,a+1)$} (q1)
				(q1) edge[bend left] node[above] {$(a_1,-a_1)$} (q2)
				(q1) edge[bend right] node[below] {$(0,0)$} (q2)
				(q2) edge[bend left]  (d)
				(q2) edge[bend right] (d)
				(d) edge[bend left]  (q3)
				(d) edge[bend right] (q3)
				(q3) edge[bend left] node[above] {$(a_n,-a_n)$} (q4)
				(q3) edge[bend right] node[below] {$(0,0)$} (q4)
				(q4) edge node[above] {$(-a,0)$} (q5);
			\end{tikzpicture}
		\end{center}
		\caption{2-VASS to show $\NP$-hardness of coverability in dimension two.}\label{two-dim-np-hard-fig}
	\end{figure}
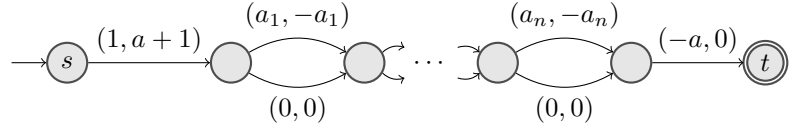
	
	\begin{proposition}\label{two-dim-np-hard-prop}
		Monus coverability in succinct $2$-VASS is $\NP$-hard.
	\end{proposition}
	\begin{proof}
		We reduce from the \emph{subset sum} problem, which is well-known to be
		$\NP$-hard. Here, we are given binary encoded numbers $a_1,\ldots,a_n,a \in \N$ and
		are asked whether there is a vector $(x_1,\ldots,x_n)\in\{0,1\}^n$ such that
		$x_1a_1+\cdots+x_na_n=a$. Given such an instance, we construct the $2$-VASS in
		\cref{two-dim-np-hard-fig}. It is clear that we can cover $t(1,1)$ from $s(0,0)$ iff the
		subset-sum instance is positive: Covering $1$ in the first counter means our
		sum is at least $a$, whereas covering $1$ in the second counter means our sum is
		at most $a$.
	\end{proof}
	
	\begin{proposition}
		Monus coverability in unary-encoded $2$-VASS is in $\NL$.
	\end{proposition}
	\begin{proof}
		This follows using the same construction as for \cref{weak-coverability}: Given
		a $2$-VASS, there are only two permutations $\sigma$ of $\{1,2\}$. Thus, we can
		try both permutations $\sigma$ and construct the VASS $\cV'_\sigma$ in
		logspace. Then, $\cV_\sigma$ has dimension $2d$. Thus, we reduce monus
		coverability in $2$-VASS to reachability in $\Z$-semantics in $4$-VASS. Since
		reachability with $\Z$-semantics in each fixed dimension can be decided in
		$\NL$~\cite{DBLP:journals/jcss/GurariI81}, this provides an $\NL$ upper bound.
	\end{proof}
	
	\section{One-dimensional VASS}\label{sec:one_dim}
	
	\subparagraph{Reachability} We begin with the proofs regarding reachability.
	\begin{proposition}\label{monus-reach-1-VASS}
		Monus reachability in $1$-VASS is in $\NL$ under unary encoding and in $\NP$ under binary encoding.
	\end{proposition}
	The proof of \cref{monus-reach-1-VASS} relies on the following simple consequence of \cref{charact-reachability}:
	\begin{lemma}\label{charact-reach-1-VASS}
		Let $\cV$ be a $1$-VASS. %
		Then $s(m) \mathrel{{\weaksteps}{}_{\cV}} t(n)$ if and only if (i)~$s(m)\steps_{\cV} t(n)$ or (ii)~  there exist a state $q$ and number
		$m'\ge m$ with $s(m')\steps_{\cV} q(0)$ and $q(0)\steps_{\cV} t(n)$.
		\end{lemma}
	
	For \cref{monus-reach-1-VASS}, we reduce to reachability in one-counter
	automata.  A \emph{one-counter automaton (OCA)} is a $1$-VASS with zero-tests, i.e.
	special transitions that test the counter for zero instead of adding a number.
	For encoding purposes, zero tests take up as much space as a transition adding
	$0$ to the counter.  In our reduction, the update encoding is
	preserved: If the input $1$-VASS has unary encoding, then the OCA has
	unary updates as well. If the input $1$-VASS has binary updates, then
	the OCA will too.  Then, we can use the fact that in OCA with unary updates,
	reachability is in $\NL$~\cite{ValiantP75}
	and for binary updates, it is in
	$\NP$~\cite{DBLP:conf/concur/HaaseKOW09}.
	
	The OCA first guesses whether to simulate a run of type~(i) or of type~(ii) in
	\cref{charact-reach-1-VASS}.  Then for type~(i), it just simulates a classical
	$1$-VASS. For type~(ii), it first non-deterministically increments the counter,
	and then simulates a run of the $1$-VASS. However, on the way, it keeps a flag
	signaling whether the counter has hit $0$ at some point (which it can maintain
	using zero tests). Thus, when simulating runs of type (ii), the OCA only
	accepts if zero has been hit. For a detailed description, refer to \cref{app:monus-reach-1-VASS}.
	
	\begin{figure}
		\begin{center}
			\begin{tikzpicture}[initial text={}]
				\node[state,initial] (q0) {$q_0$};
				\node[state,right=1.2cm of q0] (q1) {$q_1$};
				\node[state,right=1.2cm of q1] (q2) {$q_2$};
				\node[state,right=1.2cm of q2] (q3) {$q_n$};
				\node[state,right=1.2cm of q3] (q4) {$q_{n+1}$};
				\node[state,right=1.2cm of q4,accepting] (q5) {$q_f$};
				\node at ($(q2.center)!0.5!(q3.center)$) (d) {$\cdots$};
				\path[->] 
				(q0) edge node[above] {$1$} (q1)
				(q1) edge[bend left] node[above] {$a_1$} (q2)
				(q1) edge[bend right] node[below] {$0$} (q2)
				(q2) edge[bend left]  (d)
				(q2) edge[bend right] (d)
				(d) edge[bend left]  (q3)
				(d) edge[bend right] (q3)
				(q3) edge[bend left] node[above] {$a_n$} (q4)
				(q3) edge[bend right] node[below] {$0$} (q4)
				(q4) edge node[above] {$-a$} (q5);
			\end{tikzpicture}
		\end{center}
		\caption{$1$-VASS to show $\NP$-hardness of monus reachability in dimension one with binary encoded counter updates.}\label{one-dim-np-hard}
	\end{figure}
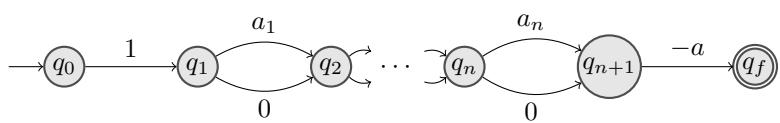
	\begin{proposition}\label{reachability-1-vass-np-hard}
		Monus reachability in $1$-VASS is $\NP$-hard under binary encoding.
	\end{proposition}
	As in \cref{two-dim-np-hard-prop}, we reduce from subset sum.  Given $a_1,\ldots,a_n,a$
	in binary, we construct the $1$-VASS in \cref{one-dim-np-hard}. Then
	$q_0(0)\weaksteps q_f(1)$ iff this is a positive instance. See \cref{proof-reachability-1-vass-np-hard}.

	\subparagraph{Zero reachability and coverability}
	\begin{proposition}
		Monus zero-reachability in $1$-VASS is in $\NL$ under unary encoding
		and in $\NC^2$ under binary encoding.
	\end{proposition}
	Since monus zero-reachability reduces to classical coverability (\cref{charact-zero-reachability}),
	this follows from existing $1$-VASS results:
	Coverability in $1$-VASS is in $\NL$ under unary encoding \cite{ValiantP75}
	and $\NC^2$ under binary encoding \cite{DBLP:conf/concur/AlmagorCPS020}.
	\begin{proposition}\label{monus-cover-1-VASS}
		Monus coverability in $1$-VASS is in $\NL$ under unary encoding and in $\NC^2$ under binary encoding.
	\end{proposition}
	The first statement follows from \cref{monus-reach-1-VASS} and the fact that monus coverability reduces to
	monus reachability by simply adding a new final state where we can count down.  For the $\NC^2$ bound, we use the
	following consequence of \cref{charact-zero-reachability} (see
  \cref{proof-charact-cover-1-vass}).
	\begin{lemma}\label{charact-cover-1-VASS}
		Let $\cV$ be a $1$-VASS with configurations $s(m)$ and $t(n)$.
		Then $t(n)$ is monus coverable from $s(m)$ in $\cV$ if and only if
		$t(n)$ is coverable from $s(m)$ in $\cV$ under classical semantics
		or there is a state $q$ of $\cV$ such that
		$t(n)$ is coverable from $q(0)$ in $\cV$ under classical semantics
		and $s(m)$ is coverable from $q(0)$ in $\rev{\cV}$ under classical semantics.
	\end{lemma}
	\begin{proof}[Proof of \cref{monus-cover-1-VASS}]
		It remains to prove the $\NC^2$ upper bound, for which we check the requirements of \cref{charact-cover-1-VASS}.
		Let $k$ be the number of states of the input $1$-VASS.
		Observe that \cref{charact-cover-1-VASS} yields a logical disjunction over $k+1$ disjuncts,
		where one disjunct consists of a single coverability check
		and the remaining $k$ each consist of a logical conjunction over two coverability checks.
		Classical coverability of binary encoded $1$-VASS is in $\NC^2$ \cite{DBLP:conf/concur/AlmagorCPS020},
		and by the definition of this complexity class, we can combine $2k + 1$ such checks according to
		the aforementioned logical relationship and still yield an $\NC^2$-algorithm.
		Note that this is only possible because $k$ is linear in the size of the input.
	\end{proof}
	
	\bibliographystyle{plain}
	\bibliography{ref}
	
	\appendix
\section{Formal Proofs for Arbitrary Dimensions}
\subsection{Reachability in Arbitrary Dimensions}

In this section, we fix a VASS $\cV = (Q, d, \Delta)$, a configuration $s_0(\vv_0) \in Q \times \N^d$, a valid sequence of transitions $\rho = (p_0, \zz_0, p_1) \ldots (p_{k-1}, \zz_{k-1}, p_k)$, and the number $\mm = \mi(\rho, s_0, \vv_0)$. \crefrange{fig:int-weak}{fig:classw} describe the relations between runs of each of the considered semantics.

\begin{figure}
	
\begin{tikzpicture}[scale=1]

	\draw[->] (-0.5,0) -- (8.5,0) node[below right] {$x$};
	\draw[->] (0,-0.5) -- (0,4.5) node[above left] {$y$};

	\draw[blue, thick, dotted, domain=0:1] plot (\x, \x);
 \draw[blue, thick, dotted, domain=1:2.5] plot (\x, -\x +2);
 \draw[blue, thick, dotted, domain=2.5:3.5] plot (\x, \x-3);
  \draw[blue, thick, dotted, domain=3.5:5] plot (\x, -2*\x+7.5);
  \draw[blue, thick, dotted, domain=5:7.5] plot (\x, 3*\x-17.5);
\draw[<->] (2.5,-0.5) -- (2.5,0)node[below right]{$m_1$};
\draw[<->] (3.5,0.5) -- (3.5,1)node[below right]{$m_1$};
\draw[<->] (5,-2.5) -- (5,0)node[below right]{$m_2$};
\draw[<->] (7.5,5) -- (7.5,7.5)node[below right]{$m_2$};
	\draw[red, thick, dashed, domain=0:1] plot (\x, \x);
 \draw[red, thick, dashed, domain=1:2] plot (\x, -\x +2);
 \draw[red, thick,dashed, domain=2:2.5] plot (\x, 0);
 \draw[red, thick,dashed, domain=2.5:3.5] plot (\x, \x-2.5);
   \draw[red, thick,dashed, domain=3.5:4] plot (\x, -2*\x+8);
  \draw[red, thick,dashed, domain=4:5] plot (\x, 0);
  \draw[red, thick,dashed, domain=5:7.5] plot (\x, 3*\x-15);

\end{tikzpicture}
\caption{Figure depicting relation between runs in monus and $\Z$ semantics} 
    \label{fig:int-weak}
\end{figure}
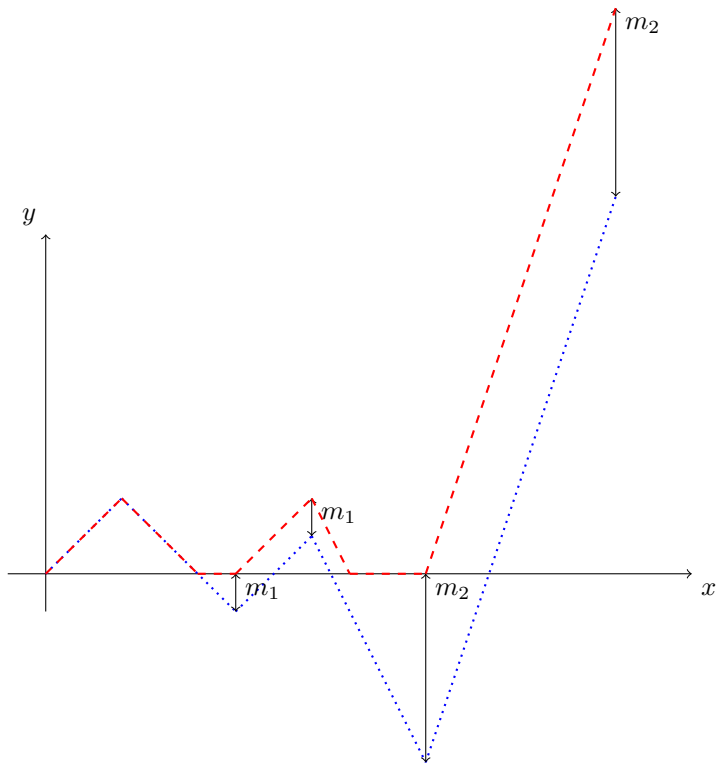
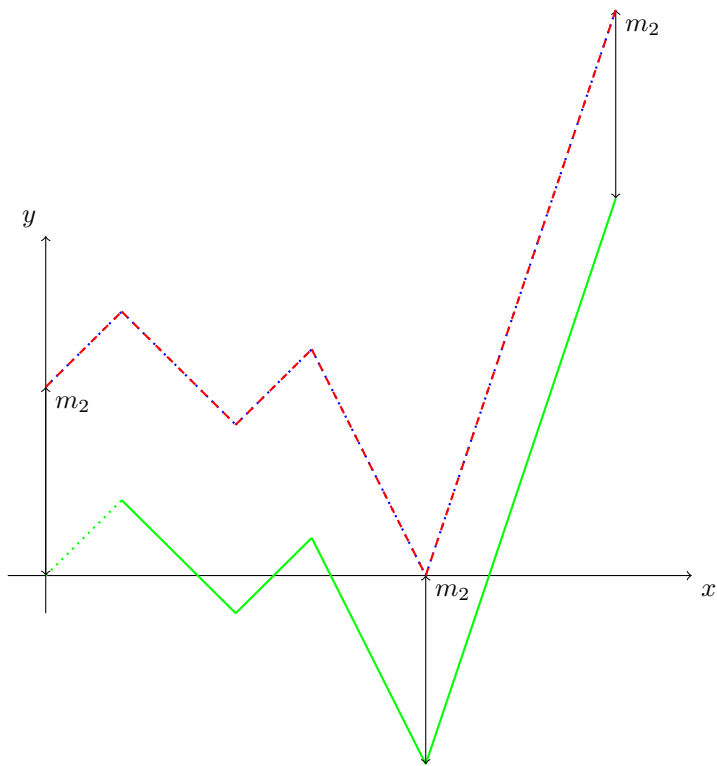
\begin{figure}
   
\begin{tikzpicture}[scale=1]

	\draw[->] (-0.5,0) -- (8.5,0) node[below right] {$x$};
	\draw[->] (0,-0.5) -- (0,4.5) node[above left] {$y$};

	\draw[green, thick, dotted, domain=0:1] plot (\x, \x);
 \draw[green, thick, domain=1:2.5] plot (\x, -\x +2);
 \draw[green, thick, domain=2.5:3.5] plot (\x, \x-3);
  \draw[green, thick, domain=3.5:5] plot (\x, -2*\x+7.5);
  \draw[green, thick, domain=5:7.5] plot (\x, 3*\x-17.5);

	\draw[blue, thick, dotted, domain=0:1] plot (\x, \x+2.5);
 \draw[blue, thick, dotted, domain=1:2.5] plot (\x, -\x +2+2.5);
 \draw[blue, thick, dotted, domain=2.5:3.5] plot (\x, \x-3+2.5);
  \draw[blue, thick, dotted, domain=3.5:5] plot (\x, -2*\x+7.5+2.5);
  \draw[blue, thick, dotted, domain=5:7.5] plot (\x, 3*\x-17.5+2.5);
\draw[<->] (0,0) -- (0,2.5)node[below right]{$m_2$};
\draw[<->] (5,-2.5) -- (5,0)node[below right]{$m_2$};
\draw[<->] (7.5,5) -- (7.5,7.5)node[below right]{$m_2$};
	\draw[red, thick, dashed, domain=0:1] plot (\x, \x+2.5);
 \draw[red, thick, dashed, domain=1:2.5] plot (\x, -\x +2+2.5);
 \draw[red, thick,dashed, domain=2.5:3.5] plot (\x, \x-3+2.5);
   \draw[red, thick,dashed, domain=3.5:5] plot (\x, -2*\x+7.5+2.5);

  \draw[red, thick,dashed, domain=5:7.5] plot (\x, 3*\x-15);

\end{tikzpicture}
	
    \caption{Figure depicting relation between monus and classical runs by adding $|\mm|$ to $\Z$-VASS.}
    \label{fig:classw}
\end{figure}

\subsubsection{Proof of \cref{prop:weak-integer}: Relating $\weak$ and $\Z$ runs}
\label{app:weak-integer}
\textbf{Statement.}
Consider the unique runs induced by $\rho$ from $s_0(\vv_0)$ in $\Z$-semantics
\[
s_0(\vv_0), \ldots, s_{k-1}(\vv_{k-1}), s_k(\vv_k),
\]
and in monus semantics
\[
s_0(\vv'_0), \ldots, s_{k-1}(\vv'_{k-1}), s_k(\vv'_k).
\]
where $\vv'_0 = \vv_0$. Then $\vv'_k = \vv_k - \mm$.
\begin{proof}
  We prove the above by applying induction on length of the path i.e. $|\rho| = k$.
  For $|\rho| = k = 0$, the above proposition is trivially true. Assume that the proposition holds for any path of length $n$. Let $k = n+1$. Let $\rho = \rho_1 \rho_2$ such that $|\rho_1| = n$.
  Let $\mm_n = \mi(\rho_1, s_0, \vv_0)$. Let $j$ be any integer in $\{1,2,\ldots,n\}$.
  By induction hypothesis, $\vv'_n[j] = \vv_n[j] - \mm_n[j]$ (i). By $\Z$ semantics and $\weak$ semantics $\vv[j]_{n+1} = \vv[j]_n + \zz[j]_{n}$ and $\vv_{n+1}'[j] = \max (\vv'_{n}[j] + \zz_n[j], 0)$, respectively (ii). By (i) and (ii) we have $\vv_{n+1}'[j] = \max (\vv_{n}[j] + \zz[j] - \mm_n[j], 0) = \max(\vv_{n+1}[j] - \mm_n[j], 0)$ (iii). There are two possible cases. 
  \\Case 1: $\vv_{n+1}[j] \ge \mm_n[j]$ (iv). Then $\mm[j] = \mm_n[j]$ (by definition of $\mi$) (v). Thus we have $\vv_{n+1}'[j] \overset{\text{(iii)}}{=} \max(\vv_{n+1}[j] - \mm_n[j], 0) \overset{\text{(iv)}}{=} \vv_{n+1}[j] - \mm_n[j] \overset{\text{(v)}}{=} \vv_{n+1}[j] - \mm[j]$. 
  \\Case 2: $\vv_{n+1}[j] < \mm_n[j]$ (vi). Then $\mm[j] = \vv_{n+1}[j]$ (by definition of $\mi$) (vii). Thus we have $\vv_{n+1}'[j] \overset{\text{(iii)}}{=} \max(\vv_{n+1}[j] - \mm_n[j], 0) \overset{\text{(vi)}}{=} 0 \overset{\text{(vii)}}{=} \vv_{n+1}[j] - \mm[j]$.
  \\Hence, $\vv_{n+1}' = \vv_{n+1} - \mm$.
\end{proof}

\subsubsection{Proof of \cref{prop:weak-class}: Relating $\weak$ and Classical runs}
\label{app:weak-class}
\textbf{Statement}- 
Consider the following unique run corresponding to the path $\rho$ from $s_0(\vv_0)$ in the monus semantics
\[
s_0(\vv_0), \ldots, s_{k-1}(\vv_{k-1}), s_k(\vv_k).
\] Then the following run, induced by $\rho$, exists in the classical semantics
\[
s_0(\vv'_0), \ldots, s_{k-1}(\vv'_{k-1}), s_k(\vv'_k).
\]
where $\vv'_0 = \vv_0 - \mm$ and $\vv'_k = \vv_k$.
\begin{proof}
By \cref{prop:weak-integer}, $s_0(\vv_0), \ldots, s_k(\vv_k)$ is a run corresponding to $\rho$ in monus semantics iff $s_0(\vv''_0) \ldots s_k(\vv''_k)$ is the run corresponding to $\rho$ in $\Z$ semantics where $\vv_0'' = \vv_0$ and $\vv''_k = \vv_k + \mm$. By \cref{rem:zvass}, $R = s_0(\vv_0''-\mm), s_1(\vv''_1-\mm)\ldots s_k(\vv''_k-\mm)$ is a valid $\Z$ run on the same path $\rho$. Notice that, for any counter $1 \le j \le d$, $\mm[j]$ is 0 iff $j$ never goes negative in $s_0(\vv''_0) \ldots s_k(\vv''_k)$. Otherwise, $\mm[j]$ is the lowest number that the counter $j$ reaches in the run $s_0(\vv''_0) \ldots s_k(\vv''_k)$. Hence, all the vectors, $\vv_0'' - \mm, \ldots \vv_k'' - \mm$ are in $\N^d$. Therefore none of the coordinates in the run $R$ goes negative. By \cref{rem:vass}, $R = s_0(\vv_0''-\mm), s_1(\vv''_1-\mm)\ldots s_k(\vv''_k-\mm)$ is a valid classical run corresponding to $\rho$. Finally, $\vv_0'' - \mm = \vv_0 - \mm$ and $\vv''_k - \mm = \vv_k$. This concludes the proof.
\end{proof}

\subsubsection{Proof of \cref{charact-reachability}: Characterizing $\weak$ Reachability}
\label{app:charact-reachability}
\textbf{Statement}- 
Let $\cV=(d,Q,\Delta)$ be a VASS, let $s(\vv)$ and $t(\ww)$ be
configurations of $\cV$, and let $\rho$ be a path of $\cV$. Then, $s(\vv) \xweaksteps{\rho} t(\ww)$ if and only
if there is a subset $Z\subseteq\{1,\ldots, d\}$ and a vector $\vv' \ge \vv$ such that
\begin{enumerate}
  \item $s(\vv') \xsteps{\rho} t(\ww)$,
  \item For every $z\in Z$, the coordinate $z$ hits $0$ in $s(\vv') \xsteps{\rho} t(\ww)$,
  \item For every $j\in\{1,\ldots,d\}\setminus Z$, we have $\vv'[j]=\vv[j]$.
\end{enumerate}
\begin{proof}
Let $\mm = \mi(\rho, p, \vv)$. Intuitively, the $(\implies)$-direction is implied by \cref{prop:weak-class} along with the following argument. Any counter $j \in \{1,\ldots, d\}$   hits $0$ in $s(\vv) \xweaksteps{\rho} t(\ww)$ iff it hits $0$ in $s(\vv - \mm) \xsteps{\rho} t(\ww)$. Moreover, if $j$ doesn't hit $0$ in  $s(\vv) \xweaksteps{\rho} t(\ww)$ then $\mm[j] = 0$.  
Formally, by \cref{prop:weak-integer,prop:weak-class,rem:zvass}, $s(\vv) \xweaksteps{\rho} t(\ww)$ iff $s(\vv) \xzsteps{\rho} t(\ww+\mm)$ iff $s(\vv-\mm) \xzsteps{\rho} t(\ww)$ iff $s(\vv-\mm) \xsteps{\rho} t(\ww)$. Notice that if a coordinate $j$ doesn't hit $0$ in $s(\vv) \xweaksteps{\rho} t(\ww)$ then it doesn't in $s(\vv) \xzsteps{\rho} t(\ww+\mm)$ (the counters that don't hit 0 have identical behaviour in all the semantics discussed, by definition). This implies $j$ doesn't hit $0$ in $s(\vv) \xzsteps{\rho} t(\ww+\mm)$. Hence, it doesn't in $s(\vv-\mm) \xzsteps{\rho} t(\ww)$ (as $-\mm \in \N^d$) and therefore, it deosn't in $s(\vv-\mm) \xweaksteps{\rho} s(\ww)$ either (as the runs are identical by \cref{rem:vass-approxim}).  Moreover, by definition of $\mi$, $\mm[j] = 0$, if $j$ doesn't hit $0$ in $s(\vv) \xweaksteps{\rho} t(\ww + \mm)$. Hence, For all the dimensions $j$ that don't hit $0$ in $s(\vv-\mm) \xsteps{\rho} t(\ww)$ (i.e., $j \notin Z$), we have $(\vv-\mm)[j] = \vv[j]$. For all the dimensions $j$ that hit $0$ (i.e., $j \in Z$) we have $(\vv-\mm)[j] \ge  \vv[j]$ (as $-\mm \in \N^d$).
With the choice of $\vv' = \vv - \mm$, we conclude the proof of this direction.

For the $(\impliedby)$-direction, let $\vv'$ be any vector such that $\vv' \ge \vv$ and $s(\vv') \xsteps{\rho} t(\ww)$. Then, if any coordinate $i \in \{1 \ldots d\}$ doesn't hit $0$ in $s(\vv') \xsteps{\rho} t(\ww)$, we have $\vv'[i] = \vv[i]$. Let $s(\vv) \xweaksteps{\rho} t(\ww'')$. It suffices to show that $\ww = \ww''$.

Let $s(\vv') \xsteps{\rho} t(\ww)$ and $s(\vv) \xweaksteps{\rho} t(\ww'')$ be runs of the form $p_0( \vv_0) \ldots p_k( \vv_k)$ and \linebreak $p_0( \vv'_0) \ldots p_k( \vv'_k)$, respectively, where $p_0 = s$, $p_k = t$, $\vv_0 = \vv$, $\vv'_0 = \vv'$, $\vv_k = \ww''$ and $\vv'_k = \ww$. Since the same sequence of transitions is applied in both of these runs, $\vv'\ge \vv$ implies $\vv_1 \le \vv'_1$, $\vv_2 \le \vv'_2$, \ldots, $\vv_k \le \vv'_k$ (*). Moreover, if  a coordinate $i$ never hits $0$ in the former (classical run) then both the runs agree on the value of the counter throughout, as $\vv[j] = \vv'[i]$. Otherwise if $\vv_j[i] = 0$ for for some $1 \le j \le k$ then $\vv_j'[i] = 0 = \vv[i]$ by (*). Hence, from this point onwards, both the runs agree on the value of $i$. Therefore both the runs end up in the same configuration. 

Formally, for any coordinate $i \in \{1 \ldots d\}$ we distinguish two cases. (Case 1) $\vv_j[i] > 0$ for every $0 \le j \le k$ and $\vv' [i] = \vv [i]$. Then for every $0 \le j \le k$, $\vv'_j[i] = \vv_j[i]$ (as we apply an identical sequence of transitions). Hence, $\vv'_k[i] = \vv_k[i]$. (Case 2) There exists $0 \le j \le k$ such that $\vv_j[i] = 0$ and $\vv[i] \le \vv'[i]$. Without loss of generality, we assume that $j$ is the smallest such number. Notice that, until the $(j-1)$th step we have $\vv_0[i] \le \vv'_0[i]$, $\vv_1[i] \le \vv'_1[i]$, $\vv_2[i] \le \vv'_2[i], \ldots \vv_{j-1}[i] \le \vv'_{j-1}[i]$ ($\dagger$). Let $\rho[j] = (p_{j-1}, \zz_{j-1}, p_j)$.
As $\vv_j[i] = 0$, $\zz_{j-1}[i] = - \vv_{j-1}[i]$ ($\ddagger$). By definition of monus semantics $\vv'_{j}[i] = \max (0, \vv'_{j-1}[i] - \zz_{j-1}[i])$. By ($\dagger$) and ($\ddagger$) we have $\vv'_{j-1}[i] - \zz_{j-1}[i] \le 0$. Hence, $\vv'_{j}[i] = \vv_j[i] = 0$. Similar to Case 1, as the sequence of transitions in both the runs are the same, the value of counter $i$ in both the runs synchronizes from step $j$ onwards.
\end{proof}

\section{Formal Proofs for Two-dimensional VASS}

\subsection{Reachability for Two-dimensional VASS}

\subsubsection{Proof of \cref{weak-reach-2-VASS-upper}}\label{proof-reachability-2-VASS-upper}

\textbf{Statement}- In succinct 2-VASS, reachability with monus semantics is in $\PSPACE$.
\begin{proof}
  We reduce to the constrained runs problem for succinct $2$-VASS.
  Let $\cV = (2,Q,\Delta)$ be a $2$-VASS with configurations $p(\vv)$ and $q(\ww)$.
  By \cref{charact-reachability} existence of a monus run $p(\vv) \xweaksteps{\rho} q(\ww)$
  for some path $\rho$ is equivalent to the following:
  There is a subset $Z=\{z_1,\ldots,z_k\}\subseteq \{1,2\}$ and there are configurations
  $p_1(\vv_1),p_2(\vv_2)$, $p(\vv')$ so that
  \begin{enumerate}
    \item $p(\vv')\steps p_1(\vv_1)\steps p_2(\vv_2)\steps q(\ww)$ on the classical run induced by $\rho$,
    \item for each $j\in[1,k]$, we have $\vv_j[z_j]=0$,
    \item $\vv'\ge\vv$
    \item for each $i\in[1,d]\setminus Z$, we have $\vv'[i]=\vv[i]$.
  \end{enumerate}
  For each choice of $p_1,p_2 \in Q$ here, we translate these conditions into an instance of
  the constrained runs problem, which we can then solve in $\PSPACE$ by \cref{constrained-runs-PSPACE}.
  Since enumerating all possibilities for $p_1,p_2$ is also possible in $\PSPACE$, this results in
  a $\PSPACE$ algorithm as required.
  We simply accept if one of the possibilities results in a positive instance of the constrained runs problem, and reject otherwise.
  
  For the translation to a constrained runs problem instance,
  we first construct a new $2$-VASS $\cV'$ from $\cV$ by making all states final,
  adding a new initial state $q_0$,
  and adding the transitions $(q_0,(1,0),q_0)$, $(q_0,(0,1),q_0)$, and $(q_0,\vv,p)$.
  This ensures that condition (1) becomes
  $q_0(0,0) \steps p(x_1,y_1)\steps p_1(x_2,y_2)\steps p_2(x_3,y_3)\steps q(x_4,y_4)$ in $\cV'$,
  as required by the constrained runs problem.
  The added loops on $q_0$ furthermore ensure that we can reach $p(\vv')$ for any $\vv' \geq \vv$.
  Secondly, we fix the numbers $m = 4 = s = t$ and the sequence of states
  $q_0, q_1 = p, q_2 = p_1, q_3 = p_2, q_4 = q$.
  Thirdly, we use the Presburger formula $\ww[1] = x_4 \wedge \ww[2] = y_4$ to express
  that the run ends in the desired configuration $q(\ww)$.
  Finally, we need to express conditions (2) to (4) using equivalent quantifier-free Presburger formulas.
  We will use the conjunction of all the constructed formulas as input formula $\psi$ for the constrained runs problem.
  Condition (3) simply becomes $\vv[1] \leq x_1 \wedge \vv[2] \leq y_1$, where $a \leq b$ is syntactic sugar for $a<b \vee a=b$.
  For conditions (2) and (4) we can go through all five possibilities for the set $Z$
  and construct a big disjunction:
  \begin{align*}
    (\vv[1] = x_1 \wedge \vv[2] = y_1) &\vee (0 = x_2 \wedge \vv[2] = y_1) \vee (\vv[1] = x_1 \wedge 0 = y_2)\\
    &\vee (0 = x_2 \wedge 0 = y_3) \vee (0 = x_3 \wedge 0 = y_2)
  \end{align*}
  Note that the order of elements in $Z$ matters,
  as $z_1 = 1, z_2 = 2$ and $z_1 = 2, z_2 = 1$ result in two different conditions here.
  
  Clearly the conditions (1) to (4) are equivalent to existence of a constrained run in $\cV'$ subject to the conjunction of the constructed formulas.
\end{proof}
In the above proof the constructed algorithm enumerates all possible choices of states $p_1, p_2$,
and moreover goes through all possibilities for the set $Z$ and a total order over its elements.
We remark that since $\NPSPACE = \PSPACE$, one could alternatively have the algorithm make
nondeterministic guesses in both these cases.
The resulting nondeterministic algorithm would have been sufficient to show $\PSPACE$-membership.

\section{Formal proofs for One-dimensional VASS}

\subsection{Reachability for One-dimensional VASS}

\subsubsection{Proof of \cref{monus-reach-1-VASS}}
\label{app:monus-reach-1-VASS}

\textbf{Statement}-
Monus reachability in $1$-VASS is in $\NL$ under unary encoding and in $\NP$ under binary encoding.

\begin{proof}
  We reduce this problem to classical configuration reachability in one-counter
  automata, which is in $\NL$ for unary counter updates and in $\NP$ for binary
  counter updates~\cite{DBLP:conf/concur/HaaseKOW09}.
  
  Consider an instance of the monus reachability problem in $1$-VASS:
  Let $\cV$ be a $1$-VASS with configurations $s(m)$ and $t(n)$.
  In the following we construct a one-counter automaton $\cA$ with two configurations
  in such a way that reachability will be equivalent to the characterization given by \cref{charact-reach-1-VASS}.
  To this end we start with three copies of $\cV$, which we call $\cV_0, \cV_1, \cV_2$.
  The first copy, $\cV_0$, is supposed to check classical reachability of $t(n)$ from $s(m)$ in $\cV$,
  whereas the other two copies handle other case of \cref{charact-reach-1-VASS}.
  We add a new state $s'$ and a gadget that from $s'$ either jumps to $s$ in $\cV_0$,
  or first increments the counter arbitrarily (via a loop) and then jumps to $s$ in $\cV_1$.
  Every state in $\cV_1$ is also connected to the same state in $\cV_2$ with a zero test.
  Then finally we add a new state $t'$ and ensure that starting from $t$ in either $\cV_0$ or $\cV_2$,
  one can jump to $t'$ without changing the counter.
  This completes the construction fo $\cA$; as its two configurations we choose $s'(m)$ and $t'(n)$.
  Note that the size of $\cA$ is linear in the size of $\cV$.
  
  Let us now argue why classical reachability from $s(m)$ to $t(n)$ in $\cV$.
  is equivalent to monus reachability from $s'(m)$ to $t'(n)$ in $\cA$.
  By \cref{charact-reach-1-VASS} the former is equivalent to (a) classical reachability in $\cV$
  or (b) existence of a state $q$ such that $t(n)$ is classically reachable form $q(0)$
  and $s(m)$ is classically coverable from $q(0)$ in the reverse VASS.
  The second part of case (b) can be equivalently restated as existence of a counter value $\ell \geq m$
  such that $q(0)$ is classically reachable from $s(\ell)$ in $\cV$.
  Both parts of case (b) together thus are equivalent to existence of a run $\rho$
  from $s(\ell)$ to $t(n)$ under classical semantics such that
  $\ell > m$ and $\rho$ reaches counter value $0$ in some state $q$.
  It is not difficult to see that reachability from $s'(m)$ to $t'(n)$ in $\cA$ by going through
  $\cV_1$ and $\cV_2$ is equivalent to case (b), whereas going through $\cV_0$ is equivalent to case (a).
  Since $t'$ can only be reached from $s'$ in $\cA$ by going through these copies of $\cV$,
  we thus have proven the reduction correct.
\end{proof}

\subsubsection{Proof of \cref{reachability-1-vass-np-hard}}\label{proof-reachability-1-vass-np-hard}

\textbf{Statement}-
Monus reachability in $1$-VASS is $\NP$-hard under binary encoding.

\begin{proof}
  Given $a_1,\ldots,a_n,a$ in binary, we construct the $1$-VASS in
  \cref{one-dim-np-hard}. It is clear from the construction that in this
  $1$-VASS, we can monus reach $q_f(1)$ from $q_0(0)$ if and only if the
  subset-sum instance is positive:
  If we reach $q_f(1)$ from $q_0(0)$ under classical semantics, then it is
  clear that our sum equates to exactly $a$.
  On the other hand, the only possible subtraction in this $1$-VASS is the transition
  $q_{n+1} \xrightarrow{-a} q_f$, which, if performed in monus semantics, would lead to $q_f(0)$ instead of $q_f(1)$.
  Therefore all runs reaching the latter configuration are also valid under classical semantics.
  This completes the proof.
\end{proof}

\subsection{Coverability for One-dimensional VASS}

\subsubsection{Proof of \cref{charact-cover-1-VASS}}\label{proof-charact-cover-1-vass}

\textbf{Statement}-
Let $\cV$ be a $1$-VASS with configurations $s(m)$ and $t(n)$.
Then $t(n)$ is monus coverable from $s(m)$ in $\cV$ if and only if
$t(n)$ is coverable from $s(m)$ in $\cV$ under classical semantics
or there is a state $q$ of $\cV$ such that
$t(n)$ is coverable from $q(0)$ in $\cV$ under classical semantics
and $s(m)$ is coverable from $q(0)$ in $\rev{\cV}$ under classical semantics.

\begin{proof}
  For the only if direction, assume $t(n)$ is monus coverable from $s(m)$ in $\cV$.
  If the witnessing run $\rho$ does not reach counter value $0$ anywhere in-between,
  then clearly we also have coverability under classical semantics.
  On the other hand, let $q(0)$ be the last configuration before $t(n)$ in $\rho$ counter value $0$.
  Then by the same argument as before,
  $t(n)$ is monus coverable from $q(0)$ in $\cV$ under classical semantics.
  Moreover, $q(0)$ is monus reachable from $s(m)$ in $\cV$, which by \cref{charact-zero-reachability}
  implies that $s(m)$ is coverable from $q(0)$ in $\rev{\cV}$ under classical semantics.
  
  For the if direction, note that coverability of $t(n)$ from $s(m)$ under classical semantics
  obviously implies the same under monus semantics.
  Therefore let us assume that there is a state $q$ of $\cV$ such that
  $t(n)$ is coverable from $q(0)$ in $\cV$ under classical semantics
  and $s(m)$ is coverable from $q(0)$ in $\rev{\cV}$ under classical semantics.
  Applying \cref{charact-zero-reachability} to the latter yields monus reachability
  of $q(0)$ from $s(m)$ in $\cV$.
  If we take the witnessing run for this, and append to it the run that covers $t(n)$ from $q(0)$,
  we obtain a run that monus covers $t(n)$ from $s(m)$ as required.
\end{proof}

\end{document}